\documentclass[journal]{IEEEtran}

\usepackage[noadjust]{cite}
\usepackage{caption}
\usepackage[table]{xcolor}
\usepackage{slashbox}
\usepackage{graphicx}
\usepackage{color}
\usepackage{placeins}
\usepackage{float}
\usepackage{tabularx,colortbl}
\usepackage[cmex10]{amsmath}
\usepackage{amssymb}
\usepackage{amsbsy}
\usepackage{bm}
\usepackage{amsthm}
\usepackage{bbm}
\usepackage{epstopdf}
\usepackage{breqn}
\usepackage{cite}
\usepackage{framed}
\usepackage{url}
\usepackage{amsfonts}
\usepackage{mathrsfs}
\usepackage{mathtools}
\definecolor{purple1}{RGB}{100, 85, 255}
\definecolor{orange1}{RGB}{233, 146, 0}

\DeclareGraphicsExtensions{.eps}
\DeclareMathOperator*{\argmax}{arg\,max}
\DeclareMathOperator*{\argmin}{arg\,min}

\DeclareMathOperator*{\support}{supp}

\newtheorem{thm}{Theorem}
\newtheorem{lemma}{Lemma}
\newtheorem{corollary}{Corollary}
\newtheorem{prop}{Proposition}
\newtheorem{definition}{Definition}

\newcommand*{\QEDB}{\hfill\ensuremath{\blacksquare}}%

\begin{document}

\title{Sampling Requirements for Stable Autoregressive Estimation}

\author{Abbas~Kazemipour,~\IEEEmembership{Student Member,~IEEE,
}~Sina~Miran,~\IEEEmembership{Student Member,~IEEE,
}~Piya~Pal,~\IEEEmembership{Member,~IEEE,
}~Behtash~Babadi,~\IEEEmembership{Member,~IEEE,
}~and~Min~Wu,~\IEEEmembership{Fellow,~IEEE
}
\thanks{A. Kazemipour, S. Miran, B. Babadi and M. Wu are with the Department of Electrical and Computer Engineering (ECE), University of Maryland, College Park, MD 20742 USA (e-mails: kaazemi@umd.edu; smiran@umd.edu; behtash@umd.edu; minwu@umd.edu). P. Pal is with the Department of ECE, University of California, San Diego, La Jolla, CA 92093 (e-mail: pipal@ucsd.edu).}
\thanks{This work has been presented in part at the 50th Annual Conference on Information Sciences and Systems, 2016 \cite{kazemipour-ciss}.}
\thanks{Corresponding author: Behtash Babadi (e-mail: behtash@umd.edu).}
\thanks{Copyright (c) 2017 IEEE. Personal use of this material is permitted. However, permission to use this material for any other purposes must be obtained from the IEEE by sending an email to pubs-permissions@ieee.org.}}

\maketitle
%\thispagestyle{plain}
%\pagestyle{plain}
%\graphicspath{ {./figures/} }

\begin{abstract}
We consider the problem of estimating the parameters of a linear univariate autoregressive model with sub-Gaussian innovations from a limited sequence of consecutive observations. Assuming that the parameters are compressible, we analyze the performance of the $\ell_1$-regularized least squares as well as a greedy estimator of the parameters and characterize the sampling trade-offs required for stable recovery in the non-asymptotic regime. In particular, we show that {for a fixed sparsity level,} stable recovery of AR parameters is possible when the number of samples scale \emph{sub-linearly} with the AR order. Our results improve over existing sampling complexity requirements in AR estimation using the LASSO, when the sparsity level scales faster than the square root of the model order. We further derive sufficient conditions on the sparsity level that guarantee the minimax optimality of the $\ell_1$-regularized least squares estimate. Applying these techniques to simulated data as well as real-world datasets from crude oil prices and traffic speed data confirm our predicted theoretical performance gains in terms of estimation accuracy and model selection.
\end{abstract}

% keywords

% \fi
\IEEEpeerreviewmaketitle


% keywords

% \fi
\IEEEpeerreviewmaketitle

\section{Introduction}
Autoregressive (AR) models are among the most fundamental tools in analyzing time series. {Applications include financial time series analysis \cite{sang2015simultaneous} and traffic modeling \cite{farokhi2014vehicular,ahmed1982application,ahmed1979analysis,barcelo2010travel,clark2003traffic,robinson2003time}.} Due to their well-known approximation property, these models are commonly used to represent stationary processes in a parametric fashion and thereby preserve  the underlying structure of these processes \cite{akaike1969fitting}. In order to leverage the approximation property of AR models, often times parameter sets of very large order are required \cite{poskitt2007autoregressive}. For instance, any autoregressive moving average (ARMA) process can be represented by an AR process of infinite order. Statistical {inference} using these models is usually performed through fitting a long-order AR model to the data, which can be viewed as a truncation of the infinite-order representation \cite{shibata1980asymptotically, galbraith1997some, galbraith2001autoregression, ing2005order}. In general, the ubiquitous long-range dependencies in real-world time series, such as financial data, results in AR model fits with large orders \cite{sang2015simultaneous}.

In various applications of interest, the AR parameters fit to the data exhibit sparsity, that is, only a small number of the parameters are non-zero. Examples include autoregressive communication channel models, quasi-oscillatory data tuned around specific frequencies and financial time series \cite{baddour2005autoregressive, mann1999oscillatory, robinson2003time}. The non-zero AR parameters in these models correspond to significant time lags at which the underlying dynamics operate. {Traditional AR order selection criteria such as the Final Prediction Error (FPE) \cite{akaike1973maximum}, Akaike Information Criterion (AIC) \cite{akaike1970statistical} and Bayesian Information Criterion (BIC) \cite{schwarz1978estimating}, are based on asymptotic lower bounds on the mean squared prediction error.  Although there exist several improvements over these traditional results aiming at exploiting sparsity \cite{ing2005order,shibata1980asymptotically,wang2007regression}, the resulting criteria pertain to the asymptotic regimes and their finite sample behavior is not well understood  \cite{goldenshluger2001nonasymptotic}. Non-asymptotic results for AR estimation, such as \cite{goldenshluger2001nonasymptotic,nardi2011autoregressive}, do not fully exploit the sparsity of the underlying parameters in favor of reducing the sample complexity. In particular, for an AR process of order $p$, sufficient sampling requirements of $n \sim \mathcal{O}(p^4) \gg p$ and $n \sim \mathcal{O}(p^5) \gg p$ are established in \cite{goldenshluger2001nonasymptotic} and \cite{nardi2011autoregressive}, respectively.}  

{A relatively recent line of research employs the theory of compressed sensing (CS) for studying non-asymptotic sampling-complexity trade-offs for regularized M-estimators.} In recent years, the CS theory has become the standard framework for measuring and estimating sparse statistical models \cite{donoho2006compressed, candes2006compressive, candes2008introduction}. The theoretical guarantees of CS imply that when the number of incoherent measurements are roughly proportional to the sparsity level, then stable recovery of the model parameters is possible. {A key underlying assumption in many existing theoretical analyses of linear models is the independence and identical distribution (i.i.d.) of the covariates' structure. The matrix of covariates is either formed by fully i.i.d. elements \cite{rudelson2008sparse,baraniuk2008simple}, is based on row-i.i.d. correlated designs \cite{zhao2006model,raskutti2010restricted}, is Toeplitz-i.i.d. \cite{Toeplitz}, {or circulant i.i.d. \cite{rauhut2012restricted}}, where the design is extrinsic, fixed in advance and is independent of the underlying sparse signal. The matrix of covariates formed from the observations of an AR process does not fit into any of these categories, as the intrinsic history of the process plays the role of the covariates. Hence the underlying interdependence in the model hinders a straightforward application of existing CS results to AR estimation.} {Recent non-asymptotic results on the estimation of multi-variate AR (MVAR) processes have been relatively successful in utilizing sparsity for such dependent structures. For Gaussian and low-rank MVAR models, respectively, sub-linear sampling requirements have been established in \cite{loh2012high,han2013transition} and \cite{negahban2011estimation}, using regularized LS estimators, under bounded operator norm assumptions on the transition matrix. These assumptions are shown to be restrictive for MVAR processes with lags larger than 1 \cite{wong2016regularized}. By relaxing these boundedness assumptions for Gaussian, sub-Gaussian and heavy-tailed MVAR processes, respectively, sampling requirements of $n \sim \mathcal{O}(s \log p)$ and $\mathcal{O}((s \log p)^2)$ have been established in \cite{basu2015regularized} and \cite{ wong2016regularized, wu2016performance}. However, the quadratic scaling requirement in the sparsity level for the case of sub-Gaussian and heavy-tailed innovations incurs a significant gap with respect to the optimal guarantees of CS (with linear scaling in sparsity), particularly when the sparsity level $s$ is allowed to scale with $p$.}

{In this paper, we consider two of the widely-used estimators in CS, namely the $\ell_1$-regularized Least Squares (LS) or the LASSO and the Orthogonal Matching Pursuit (OMP) estimator, and extend the non-asymptotic recovery guarantees of the CS theory to the estimation of univariate AR processes with compressible parameters using these estimators. In particular, we improve the aforementioned gap between non-asymptotic sampling requirements for AR estimation and those promised by compressed sensing by providing sharper sampling-complexity trade-offs which improve over existing results when the sparsity grows faster than the square root of $p$. Our focus on the analysis of univariate AR processes is motivated by the application areas of interest in this paper which correspond to one-dimensional time series. Existing results in the literature \cite{loh2012high,han2013transition,negahban2011estimation,basu2015regularized,wong2016regularized, wu2016performance}, however, consider the MVAR case and thus are broader in scope. We will therefore compare our results to the univariate specialization of the aforementioned results.} Our main contributions can be summarized as follows:

First, we establish that {for a univariate AR process with sub-Gaussian innovations} when the number of measurements scales \emph{sub-linearly} with the product of the ambient dimension $p$ and the sparsity level $s$, i.e., $n \sim \mathcal{O}(s  (p \log p)^{1/2}) \ll p$, then stable recovery of the underlying AR parameters is possible using the LASSO and the OMP estimators, even though the covariates are highly interdependent and solely based on the history of the process. In particular, when $s \propto p^{\frac{1}{2} + \delta}$ for some $\delta \ge 0$ and the LASSO is used, our results improve upon those of \cite{wong2016regularized, wu2016performance}, when specialized to the univariate AR case, by a factor of $p^{\delta} (\log p)^{{3}/{2}}$. For the special case of Gaussian AR processes, stronger results are available which require a scaling of $n \sim \mathcal{O}(s \log p)$ \cite{basu2015regularized}. Moreover, our results provide a theory-driven choice of the number of iterations for stable estimation using the OMP algorithm, which has a significantly lower computational complexity than the LASSO.

Second, in the course of our analysis, we establish the Restricted {Eigenvalue} (RE) condition \cite{bickel2009simultaneous} for $n \times p$ design matrices formed from a realization of an AR process in a Toeplitz fashion, when $n \sim \mathcal{O}(s  (p \log p)^{1/2}) \ll p$. To this end, we invoke appropriate concentration inequalities for sums of \emph{dependent} random variables in order to capture and control the high interdependence of the design matrix. {In the special case of a white noise sub-Gaussian process, i.e., a sub-Gaussian i.i.d. Toeplitz measurement matrix, we show that our result can be strengthened from $n \sim \mathcal{O}(s (p \log p)^{1/2})$ to $n \sim \mathcal{O}(s (\log p)^2)$, which improves by a factor of $s / \log p$ over the results of \cite{Toeplitz} requiring $n \sim \mathcal{O} (s^2 \log p)$.}

Third, we establish sufficient conditions on the sparsity level which result in the minimax optimality of the $\ell_1$-regularized LS estimator. Finally, we provide simulation results as well as application to oil price and traffic data which reveal that the sparse estimates significantly outperform traditional techniques such as the Yule-Walker based estimators \cite{stoica1997introduction}. We have employed statistical tests in time and frequency domains to compare the performance of these estimators.

The rest of the paper is organized as follows. In Section \ref{sec:ar_notations}, we will introduce the notations and problem formulation. In Section \ref{sec:ar_theory}, we will describe several methods for the estimation of the parameters of an AR process, present the main theoretical results of this paper on robust estimation of AR parameters, and establish the minimax optimality of the $\ell_1$-regularized LS estimator. Section \ref{sec:ar_sim} includes our simulation results on simulated data as well as the real-world financial and traffic data, followed by concluding remarks in Section \ref{sec:conc}.

\section{Notations and Problem Formulation}
\label{sec:ar_notations}
Throughout the paper we will use the following notations. We will use the notation $\mathbf{x}_i^j$ to denote the vector $[x_i,\cdots,x_j]^T$. We will denote the estimated values by $\widehat{(.)}$ and the biased estimates with the superscript $(.)^b$.  Throughout the proofs, $c_i$'s express absolute constants which may change from line to line where there is no ambiguity. By $c_\eta$ we mean an absolute constant which only depends on a positive constant $\eta$.
%{For a matrix $\mathbf{B}$ we denote its $j$-th column and row by $\mathbf{B}^{(j)}$ and $\mathbf{B}_{(j)}$ respectively. For a vector $\mathbf{b}$ we denote it's $j$-th element by $(\mathbf{b})_j$.}

Consider a univariate AR($p$) process defined by
\begin{equation}
\label{AR_def}
x_k = \theta_1 x_{k-1} + \theta_2 x_{k-2} + \cdots + \theta_p x_{k-p} + w_k =  \boldsymbol{\theta}^T \mathbf{x}_{k-p}^{k-1} +w_k,
\end{equation}
where $\{w_k\}_{k=-\infty}^{\infty}$ is an i.i.d sub-Gaussian innovation sequence with zero mean and variance $\sigma^2_{\sf w}$. This process can be considered as the output of an LTI system with transfer function
\begin{equation}
\label{eq:ar_tf}
H(z) = \frac{\sigma^2_{\sf w}}{1-\sum_{\ell=1}^p \theta_\ell z^{-\ell}}.
\end{equation}

Throughout the paper we will assume $\|\boldsymbol{\theta}\|_1 \leq 1-\eta <1$ to enforce the stability of the filter. We will refer to this assumption as \emph{the sufficient stability assumption}, since an AR process with poles within the unit circle does not necessarily satisfy $\| \boldsymbol{\theta} \|_1 < 1$. However, beyond second-order AR processes, it is not straightforward to state the stability of the process in terms of its parameters in a closed algebraic form, which in turn makes both the analysis and optimization procedures intractable.  {As we will show later, the only major requirement of our results is the boundedness of the spectral spread (also referred to as condition number) of the AR process. Although the sufficient stability condition is more restrictive, it will significantly simplify the spectral constants appearing in the analysis and clarifies the various trade-offs in the sampling bounds (See, for example, Corollary  \ref{cor:eig_conv}). }

The AR($p$) process given by $\{x_k\}_{k=-\infty}^\infty$ in (\ref{AR_def}) is stationary in the strict sense. Also by (\ref{eq:ar_tf}) the power spectral density of the process equals
\begin{equation}
\label{psd}
S(\omega)= \frac{\sigma_{\sf w}^2}{|1-\sum_{\ell=1}^p \theta_\ell e^{-j \ell \omega}|^2}.
\end{equation}

The sufficient stability assumption implies boundedness of the spectral spread of the process defined as 
\begin{equation*}
\label{def:ar_condition_nr}
\rho = \sup_{\omega} S(\omega) \Big \slash \inf_{\omega} S(\omega).
\end{equation*}
We will discuss how this assumption can be further relaxed in Appendix \ref{app:ar_main}. The spectral spread of stationary processes in general is a measure of how quickly the process reaches its ergodic state \cite{goldenshluger2001nonasymptotic}. An important property that we will use later in this paper is that the spectral spread is an upper bound on the eigenvalue spread of the covariance matrix of the process of arbitrary size \cite{haykin2008adaptive}.

We will also assume that the parameter vector $\boldsymbol{\theta}$ is compressible (to be defined more precisely later), and  can be well approximated by an $s$-sparse vector where $s \ll p$. We observe $n$ consecutive snapshots of length $p$ (a total of $n+p-1$ samples) from this process given by $\{x_k\}_{k=-p+1}^n$ and aim to estimate $\boldsymbol{\theta}$ by exploiting its sparsity; to this end, we aim at addressing the following questions {in the non-asymptotic regime}:
\begin{itemize}
\item Are the conventional LASSO-type and greedy techniques suitable for estimating $\boldsymbol{\theta}$?
\item What are the sufficient conditions on $n$ in terms of $p$ and $s$, to guarantee stable recovery?
\item Given these sufficient conditions, how do these estimators perform compared to conventional AR estimation techniques?
\end{itemize}

Traditionally, the Yule-Walker (YW) equations or least squares formulations are used to fit AR models. Since these methods
do not utilize the sparse structure of the parameters, they usually require $n \gg p$ samples in order to achieve satisfactory performance. The YW equations can be expressed as
\begin{equation}
\label{yule}
\mathbf{R} \boldsymbol{\theta} = \mathbf{r}_{-p}^{-1}, \quad r_0 = \boldsymbol{\theta}^T \mathbf{r}_{-p}^{-1} + \sigma^2_{\sf w},
\end{equation}
where $\mathbf{R}: = \mathbf{R}_{p\times p} = \mathbb{E}[\mathbf{x}_1^p \mathbf{x}_1^{pT}]$ is the $p \times p$ covariance matrix of the process and $r_k = \mathbb{E}[{x_ix_{i+k}}]$ is the autocorrelation of the process at lag $k$. The covariance matrix $R$ and autocorrelation vector $\mathbf{r}_{-p}^{-1}$ are typically replaced by their sample counterparts. Estimation of the AR($p$) parameters from the YW equations can be efficiently carried out using the Burg's method \cite{burg1967maximum}. Other estimation techniques include LS regression and maximum likelihood (ML) estimation. In this paper, we will consider the Burg's method and LS solutions as comparison benchmarks. When $n$ is comparable to $p$, these two methods are known to exhibit substantial performance differences \cite{marple1987digital}.

%Usually two distinct variants of maximum likelihood are considered: in conditional ML estimation one the objective function corresponds to the conditional distribution of last $n$ samples given the initial $p$ values in the series (this is equivalent to the least squares regression problem); in the second, the likelihood function considered is that corresponding to the unconditional joint distribution of all the values in the observed series. Substantial differences in the results of these approaches can occur if the observed series is short, or if the process is close to non-stationarity. Writing the  In this paper we will focus on the conditional ML estimation method (ML from now on) and its regularized version. Yule-Walker based methods exhibit poor performance when $n$ is small or comparable to $p$. This is mainly due to the fact that solving (\ref{yule}) requires an inversion of $\widehat{R}_p$ which might not be numerically stable. As a result the estimates will be poor. It can be shown that a \textit{biased} estimate of $r_k$'s and $R$ resolves the invertibility problem but this will be at the cost of distorting the Yule-Walker equations.

When fitted to the real-world data, the parameter vector $\boldsymbol{\theta}$ usually exhibits a degree of sparsity. That is, only certain lags in the history have a significant contribution in determining the statistics of the process. These lags can be thought of as the intrinsic delays in the underlying dynamics. To be more precise, for a sparsity level $s < p$, we denote by $S \subset \{1,2,\cdots,p \}$ the support of the $s$ largest elements of $\boldsymbol{\theta}$ in absolute value, and by $\boldsymbol{\theta}_s$ the best $s$-term approximation to $\boldsymbol{\theta}$. We also define
\begin{equation}
\sigma_s(\boldsymbol{\theta}) := \|\boldsymbol{\theta}-\boldsymbol{\theta}_s\|_1~\text{and}~\varsigma_s(\boldsymbol{\theta}) := \|\boldsymbol{\theta}-\boldsymbol{\theta}_s\|_2,
\end{equation}
which capture the compressibility of the parameter vector $\boldsymbol{\theta}$ in the $\ell_1$ and $\ell_2$ sense, respectively. Note that by definition 
$\varsigma_s(\boldsymbol{\theta}) \leq \sigma_s(\boldsymbol{\theta})$. For a fixed $\xi \in (0,1)$, we say that $\boldsymbol{\theta}$ is \emph{$(s,\xi)$-compressible} if $\sigma_s(\boldsymbol{\theta}) = \mathcal{O}(s^{1-\frac{1}{\xi}})$ \cite{needell2009cosamp} and \emph{$(s,\xi,2)$-compressible} if $\varsigma_s(\boldsymbol{\theta}) = \mathcal{O}(s^{1-\frac{1}{\xi}})$. Note that \emph{$(s,\xi,2)$-compressibility} is a weaker condition than \emph{$(s,\xi)$-compressibility} and when $\xi = 0$, the parameter vector $\boldsymbol{\theta}$ is exactly $s$-sparse.

Finally, in this paper, we are concerned with the compressed sensing regime where $n \ll p$, i.e., the observed data has a much smaller length than the ambient dimension of the parameter vector. The main estimation problem of this paper can be summarized as follows: \emph{given observations $\mathbf{x}_{-p+1}^n$ from an AR process with sub-Gaussian innovations and bounded spectral spread, the goal is to estimate the unknown $p$-dimensional $(s,\xi,2)$-compressible AR parameters $\boldsymbol{\theta}$ in a stable fashion (where the estimation error is controlled) when $n \ll p$.} 

\section{Theoretical Results}\label{theoretical}
\label{sec:ar_theory}

In this section, we will describe the estimation procedures and present the main theoretical results of this paper.

\subsection{$\ell_1$-regularized least squares estimation}
Given the sequence of observations $\mathbf{x}_{-p+1}^n$ and an estimate $\widehat{\boldsymbol{\theta}}$, the normalized estimation error can be expressed as:
\begin{equation}
\label{def:ar_L}
\mathfrak{L}\Big(\widehat{\boldsymbol{\theta}}\Big) :=  \frac{1}{n} \left \|\mathbf{x}_1^n-\mathbf{X}\widehat{\boldsymbol{\theta}}\right \|_2^2,
\end{equation}
\vspace{-.2cm}
where
\vspace{-.2cm}
\begin{equation}
\mathbf{X} = \left[ \begin{array}{cccc}
x_{n-1} & x_{n-2} & \cdots & x_{n-p}  \\
x_{n-2} & x_{n-3} & \cdots & x_{n-p-1} \\
\vdots & \vdots & \ddots & \vdots \\
x_{0} & x_{-1} & \cdots & x_{-p+1} \end{array} \right].
\end{equation}

Note that the matrix of covariates $\mathbf{X}$ is Toeplitz with highly interdependent elements. The LS solution is thus given by:

\begin{equation}
\label{eq:ar_LS}
\widehat{\boldsymbol{\theta}}_{{\sf LS}}=\argmin\limits_{\boldsymbol{\theta}\in \boldsymbol{\Theta} } \mathfrak{L}(\boldsymbol{\theta}),
\end{equation}
where 
\vspace{-.2cm}
\[\boldsymbol{\Theta} := \left \{ \boldsymbol{\theta} \in \mathbb{R}^p | \; \|\boldsymbol{\theta}\|_1 < 1- \eta \right \}\]
is the convex feasible region for which the stability of the process is guaranteed. {Note that the sufficient constraint of $\| \boldsymbol{\theta} \|_1 < 1- \eta$ is by no means necessary for stability. However, the set of all $\boldsymbol{\theta}$ resulting in stability is in general not convex. We have thus chosen to cast the LS estimator of Eq. (\ref{eq:ar_LS}) --as well as its $\ell_1$-regularized version that follows-- over a convex subset $\boldsymbol{\Theta}$, for which fast solvers exist. In addition, as we will show later, this assumption significantly clarifies the various constants appearing in our theoretical analysis. In practice, the Yule-Walker estimate is obtained without this constraint, and is guaranteed to result in a stable AR process. Similarly, for the LS estimate, this condition is relaxed by obtaining the unconstrained LS estimate and checking \emph{post hoc} for stability \cite{percival1993spectral}.} 

{Consistency of the LS estimator given by (\ref{eq:ar_LS}) was shown in \cite{sang2015simultaneous} when $n \rightarrow\infty$ for Gaussian innovations. In the case of Gaussian innovations the LS estimates correspond to conditional ML estimation and are asymptotically unbiased under mild conditions, and with $p$ fixed, the solution converges to the true parameter vector as $n \rightarrow \infty$. For fixed $p$, the estimation error is of the order $\mathcal{O}(\sqrt{p/n})$ in general \cite{Toeplitz}. However, when $p$ is allowed to scale with $n$, the convergence rate of the estimation error is not known in general.} 

%Note that the LS solution coincides with the ML estimates for gaussian innovations. and boundedness of the spectral spread of the process 
In the regime of interest in this paper, where $n \ll p$, the LS estimator is ill-posed and is typically regularized with a smooth norm. In order to capture the compressibility of the parameters, we consider the $\ell_1$-regularized LS estimator:

\begin{equation}
\label{eq:ar_lasso}
\widehat{\boldsymbol{\theta}}_{{\ell_1}}:=\argmin\limits_{\boldsymbol{\theta}\in \boldsymbol{\Theta}} \quad \mathfrak{L}(\boldsymbol{\theta})+ \gamma_n\|\boldsymbol{\theta}\|_1,
\end{equation}
where $\gamma_n > 0$ is a regularization parameter. {This estimator, deemed as the Lagrangian form of the LASSO \cite{tibshirani1996regression}, has been comprehensively studied in the sparse recovery literature \cite{wainwright2009sharp,knight2000asymptotics,meinshausen2006high} as well as AR estimation \cite{knight2000asymptotics,zhao2006model,nardi2011autoregressive,wang2007regression}. A general asymptotic consistency result for LASSO-type estimators was established in \cite{knight2000asymptotics}. Asymptotic consistency of LASSO-type estimators for AR estimation was shown in \cite{zhao2006model,wang2007regression}. For sparse models, non-asymptotic analysis of the LASSO with covariate matrices from row-i.i.d. correlated design has been established in \cite{zhao2006model,wainwright2009sharp}.}

In many applications of interest, the data correlations are exponentially decaying and negligible beyond a certain lag, and hence for large enough $p$, autoregressive models fit the data very well in the prediction error sense. {An important question is thus how many measurements are required for estimation stability? In the \emph{overdetermined} regime of $n \gg p$, the non-asymptotic properties of LASSO for model selection of AR processes has been studied in \cite{nardi2011autoregressive}, where a sampling requirement of $ n \sim \mathcal{O}(p^5)$ is established. {Recovery guarantees for LASSO-type estimators of multivariate AR parameters in the \emph{compressive} regime of $n \ll p$ are studied in \cite{loh2012high,han2013transition,negahban2011estimation,wong2016regularized,basu2015regularized,wu2016performance}. In particular, sub-linear scaling of $n$ with respect to the ambient dimension is established in \cite{loh2012high,han2013transition} for Gaussian MVAR processes and in \cite{negahban2011estimation} for low-rank MVAR processes, respectively, under the assumption of bounded operator norm of the transition matrix. In \cite{basu2015regularized} and \cite{wong2016regularized,wu2016performance}, the latter assumption is relaxed for Gaussian, sub-Gaussian, and heavy-tailed MVAR processes, respectively.} These results have significant practical implications as they will reveal sufficient conditions on $n$ with respect to $p$ as well as a criterion to choose $\gamma_n$, which result in stable estimation of $\boldsymbol{\theta}$ from a considerably short sequence of observations. The latter is indeed the setting that we consider in this paper, where the ambient dimension $p$ is fixed and the goal is to derive sufficient conditions on $n \ll p$ resulting in stable estimation.}

It is easy to verify that {the objective function and constraints in Eq.} (\ref{eq:ar_lasso}) are convex in $\boldsymbol{\theta}$ and hence $\widehat{\boldsymbol{\theta}}_{\ell_1}$ can be obtained using standard numerical solvers. Note that the solution to (\ref{eq:ar_lasso}) might not be unique. However, we will provide error bounds that hold for all possible solutions of (\ref{eq:ar_lasso}), with high probability. 

Recall that, the Yule-Walker solution is given by
\vspace{-0.5mm}
\begin{equation}
\label{eq:ar_yw}
\widehat{\boldsymbol{\theta}}_{{\sf yw}}:=\argmin\limits_{\boldsymbol{\theta}\in \boldsymbol{\Theta}} \quad \mathfrak{J}(\boldsymbol{\theta}) = \widehat{\mathbf{R}}^{-1}\widehat{\mathbf{r}}_{-p}^{-1},
\end{equation}
\vspace{-2.5mm}

\noindent where $\mathfrak{J}(\boldsymbol{\theta}):= \|\widehat{\mathbf{R}} \boldsymbol{\theta}- \widehat{\mathbf{r}}_{-p}^{-1}\|_{{2}}$. We further consider two other sparse estimators for $\boldsymbol{\theta}$ by penalizing the Yule-Walker equations. The $\ell_1$-regularized Yule-Walker estimator is defined as:
\begin{equation}
\label{eq:ar_ywl21}
\widehat{\boldsymbol{\theta}}_{{\sf yw,\ell_{2,1}}}:=\argmin\limits_{\boldsymbol{\theta}\in \boldsymbol{\Theta}} \quad \mathfrak{J}(\boldsymbol{\theta})+ \gamma_n\|\boldsymbol{\theta}\|_1,
\end{equation}
where $\gamma_n > 0$ is a regularization parameter.
Similarly, using the robust statistics instead of the Gaussian statistics, the estimation error can be re-defined as:
\begin{equation*}
\label{def:ar_J1}
\mathfrak{J}_1(\boldsymbol{\theta}):= \|\widehat{\mathbf{R}} \boldsymbol{\theta}- \widehat{\mathbf{r}}_{-p}^{-1}\|_{{1}},
\end{equation*}
we define the $\ell_1$-regularized estimates as
\begin{equation}
\label{eq:ar_ywl11}
\widehat{\boldsymbol{\theta}}_{{\sf yw,\ell_{1,1}}}:=\argmin\limits_{\boldsymbol{\theta}\in \boldsymbol{\Theta}} \quad \mathfrak{J}_1(\boldsymbol{\theta})+ \gamma_n\|\boldsymbol{\theta}\|_1.
\end{equation}

%The latter cost function corresponds to a robust estimator which shows show a slight advantage in favor of the  $\widehat{\boldsymbol{\theta}}_{{\sf yw,\ell_{2,1}}}$ estimator.

\vspace{-.3cm}
\subsection{Greedy estimation} \label{ar:subsec_greedy}
Although there exist fast solvers for the convex problems of the type given by (\ref{eq:ar_lasso}), (\ref{eq:ar_ywl21}) and (\ref{eq:ar_ywl11}), these algorithms are polynomial time in $n$ and $p$, and may not scale well with the dimension of data. This motivates us to consider greedy solutions for the estimation of $\boldsymbol{\theta}$. In particular, we will consider and study the performance of a generalized Orthogonal Matching Pursuit (OMP) algorithm \cite{OMP, zhang_omp}. A flowchart of this algorithm is given in Table \ref{tab:aromp} for completeness. At each iteration, a new component of $\boldsymbol{\theta}$ for which the gradient of the error metric $\mathfrak{f}(\boldsymbol{\theta})$ is the largest in absolute value is chosen and added to the current support. The algorithm proceeds for a total of $s^\star = \mathcal{O}(s\log s)$ steps, resulting in an estimate with $s^\star$ components. When the error metric $\mathfrak{L}(\boldsymbol{\theta})$ is chosen, the generalized OMP corresponds to the original OMP algorithm. For the choice of the YW error metric $\mathfrak{J}(\boldsymbol{\theta})$, we denote the resulting greedy algorithm by {\sf yw}OMP.
\begin{table}[h]
\vspace{-.1cm}
\begin{framed}
\vspace{-.1cm}
\small $\begin{array}{l}
\text{Input: } \mathfrak{f}(\boldsymbol{\theta}) , s^\star\\
\text{Output: } \widehat{\boldsymbol{\theta}}_{\sf OMP}=\widehat{\boldsymbol{\theta}}_{\sf OMP}^{(s^\star)}\\
\text{Initialization:}\Big\{\begin{array}{l}
\text{Start with the index set } S^{(0)}=\emptyset\\
\text{and the initial estimate }\widehat{\boldsymbol{\theta}}^{(0)}_{{\sf OMP}} = 0
\end{array}\\
\textbf{for } k=1,2,\cdots,s^\star\\
\text{  }\begin{array}{l}
j = \argmax \limits_i \left| \left( \nabla \mathfrak{f} \; \left(\widehat{\boldsymbol{\theta}}_{{\sf OMP}}^{(k-1)}\right) \right)_i\right|\\
S^{(k)}=S^{(k-1)}\cup \{j\}\\
\widehat{\boldsymbol{\theta}}_{{\sf OMP}}^{(k)} = \argmin \limits_{\support (\boldsymbol{\theta}) \subset S^{(k)}} \mathfrak{f}(\boldsymbol{\theta})
\end{array}\\
\textbf{end }\\
\end{array}$
\vspace{-.2cm}
\end{framed}
\caption{Generalized Orthogonal Matching Pursuit (OMP)}
\label{tab:aromp}
\vspace{-2mm}
\end{table}

\subsection{Estimation performance guarantees} 

The main theoretical result regarding the estimation performance of the $\ell_1$-regularized LS estimator is given by the following theorem:

\begin{thm}
\label{thm:ar_1}
If $\sigma_s(\boldsymbol{\theta}) = \mathcal{O}(\sqrt{s})$, there exist positive constants ${d_0}, d_1,d_2,d_3$ and $d_4$ such that for $n > s \max \{d_0 (\log p)^2, d_1 (p \log p)^{1/2}\}$ and a choice of regularization parameter $\gamma_n = d_2 \sqrt{\frac{\log p}{n}}$, any solution $\widehat{\boldsymbol{\theta}}_{{\ell_1}}$ to (\ref{eq:ar_lasso}) satisfies the bound
\begin{equation}
\label{eq:ar1error}
\left \|\widehat{\boldsymbol{\theta}}_{{\ell_1}}-\boldsymbol{\theta}\right \|_2 \leq d_3 \sqrt{\frac{s \log p}{n}}+ \sqrt{d_3\sigma_s(\boldsymbol{\theta})}\sqrt[4]{\frac{\log p}{n}},
\end{equation}
with probability greater than $1-\mathcal{O}(\frac{1}{n^{d_4}})$. The constants {depend on the spectral spread of the process and are} explicitly given in the proof.
\end{thm}

Similarly, the following theorem characterizes the estimation performance bounds for the OMP algorithm:
\begin{thm}
\label{thm_OMP}
If $\boldsymbol{\theta}$ is $(s,\xi,2)$-compressible for some $\xi < 1/2$, there exist positive constants ${d'_0}, d'_1,d'_{2},d'_{3}$ and $d'_4$ such that for {{$n > s \log s \max \{d'_0 (\log p)^2, d'_1 {(p \log p)^{1/2}}\}$}}, the OMP estimate satisfies the bound
\begin{equation}
\label{eq:aromp}
\left \|\widehat{\boldsymbol{\theta}}_{{\sf OMP}}-\boldsymbol{\theta}\right\|_2 \leq d'_{2} \sqrt{\frac{s \log s \log p}{n}} + d'_{3} \frac{\log s}{s^{{\frac{1}{\xi}-2}}}
\end{equation}
after $s^\star ={ 4{\rho}s \log {20 \rho s}}$ iterations with probability greater than $1-\mathcal{O}\left(\frac{1}{n^{d'_4}}\right)$. The constants {depend on the spectral spread of the process and are} explicitly given in the proof.
\end{thm}

The results of Theorems \ref{thm:ar_1} and \ref{thm_OMP} suggest that under suitable compressibility assumptions on the AR parameters, one can estimate the parameters reliably using the $\ell_1$-regularized LS and OMP estimators with much fewer measurements compared to those required by the Yule-Walker/LS based methods. To illustrate the significance of these results further, several remarks are in order:

{\noindent  \textit{\textbf{Remark 1.}} The sufficient stability assumption of $\|\boldsymbol{\theta}\|_1 \leq 1- \eta <1$ is restrictive compared to the class of stable AR models. In general, the set of parameters $\boldsymbol{\theta}$ which admit a stable AR process is not necessarily convex. This condition ensures that the resulting estimates of (\ref{eq:ar_lasso})-(\ref{eq:ar_ywl11}) pertain to stable AR processes and at the same time can be obtained by convex optimization techniques, for which fast solvers exist. A common practice in AR estimation, however, is to solve for the unconstrained problem and check for the stability of the resulting AR process \emph{post hoc}. In our numerical studies in Section \ref{sec:ar_sim}, this procedure resulted in a stable AR process in all cases. Nevertheless, the stability guarantees of Theorems \ref{thm:ar_1} and \ref{thm_OMP} hold for the larger class of stable AR processes, even though they may not necessarily be obtained using convex optimization techniques. We further discuss this generalization in Appendix \ref{app:ar_main}.
}

\noindent  \textit{\textbf{Remark 2.}} {When $\boldsymbol{\theta} = \mathbf{0}$, i.e., the process is a sub-Gaussian white noise and hence the matrix $\mathbf{X}$ is i.i.d. Toeplitz with sub-Gaussian elements, the constants $d_1$ and $d'_1$ in Theorems 1 and 2 vanish, and the measurement requirements strengthen to $n > d_0 s (\log p)^2$  and $n > d'_0 s \log s (\log p)^2$, respectively. Comparing this sufficient condition with that of \cite{Toeplitz} given by $n \sim \mathcal{O}(s^2 \log p)$ reveals an improvement of order $s (\log p)^{-1}$ by our results.}

\noindent  \textit{\textbf{Remark 3.}} {When $\boldsymbol{\theta} \neq \mathbf{0}$, the dominant measurement requirements are $n > d_1 s { (p \log p)^{1/2}}$ and $n > d'_1 s \log s {(p \log p)^{1/2}}$.} Comparing the sufficient condition $n \sim \mathcal{O}(s {(p \log p)^{1/2}})$ of Theorem \ref{thm:ar_1} with those of \cite{donoho2006compressed, candes2006compressive, candes2008introduction,wainwright2009sharp} for linear models with i.i.d. measurement matrices or row-i.i.d. correlated designs \cite{zhao2006model,raskutti2010restricted} given by $n \sim \mathcal{O}(s \log p)$  a loss of order $\mathcal{O}({(p/ \log p)^{1/2}})$ is incurred, although all these conditions require $n \ll p$. However, the loss seems to be natural as it stems from {a major difference of our setting as compared to traditional CS:} each row of the measurement matrix $\mathbf{X}$ highly depends on the entire observation sequence $\mathbf{x}_1^n$, whereas in traditional CS, each row of the measurement matrix is only related to the corresponding measurement. Hence, the aforementioned loss can be viewed as the price of self-averaging of the process accounting for the low-dimensional nature of the covariate sample space and the high inter-dependence of the covariates to the observation sequence. Recent results on M-estimation of sparse {MVAR} processes with sub-Gaussian and {heavy-tailed} innovations \cite{wu2016performance,wong2016regularized} require $n \sim \mathcal{O}(s^2 (\log p)^2)$ {when specialized to the univariate case,} which compared to our results improve the loss of $\mathcal{O}({(p/\log p)^{1/2}})$ to $(\log p)^2$ with the additional cost of quadratic requirement in the sparsity $s$. {However, in the over-determined regime of $s \propto p^{\frac{1}{2} + \delta}$ for some $\delta \ge 0$, our results imply $n \sim \mathcal{O} (p^{1 + \delta} (\log p)^{1/2})$, providing a saving of order ${p^{\delta} (\log p)^{3/2}}$ over those of \cite{wu2016performance,wong2016regularized}.}

\noindent \textit{\textbf{Remark 4.}} It can be shown that the estimation error for the LS method {in general} scales as $\sqrt{p/n}$ \cite{Toeplitz} which is not desirable when $n \ll p$. Our result, however, guarantees a much smaller error rate of the order $\sqrt{s \log p/n}$. Also, the sufficiency conditions of Theorem \ref{thm_OMP} require high compressibility of the parameter vector $\boldsymbol{\theta}$  ($\xi < 1/2$), whereas Theorem \ref{thm:ar_1} does not impose any extra restrictions on $\xi \in (0,1)$. Intuitively speaking, these two comparisons reveal the trade-off between computational complexity and measurement/compressibility requirements for convex optimization vs. greedy techniques, which are well-known for linear models \cite{bruckstein2009sparse}.

\noindent  \textit{\textbf{Remark 5.}} The condition $\sigma_s(\boldsymbol{\theta})=\mathcal{O}(\sqrt{s})$ in Theorem \ref{thm:ar_1} is not restricting for the processes of interest in this paper. This is due to the fact that the boundedness assumption on the spectral spread implies an exponential decay of the parameters (See Lemma 1 of \cite{goldenshluger2001nonasymptotic}). Finally, {the constants $d_1$, $d'_{1}$ are increasing with respect to the spectral spread of the process $\rho$. Intuitively speaking, the closer the roots of the filter given by (\ref{eq:ar_tf}) get to the unit circle (corresponding to larger $\rho$ and smaller $\eta$), the slower the convergence of the process will be to its ergodic state, and hence more measurements are required.} A similar dependence to the spectral spread has appeared in the results of \cite{goldenshluger2001nonasymptotic} for $\ell_2$-regularized least squares estimation of AR processes. 

%\noindent \textcolor{blue}{ \textit{\textbf{Remark 4.}} Theorems \ref{thm:ar_1} and \ref{thm_OMP} suggest an order selection criteria, so that for fixed $n$ and $s$, choosing $p \sim \mathcal{O}\left(\left(\frac{n}{s}\right)^{3/2}\right)$ guarantees the error bounds (\ref{eq:ar1error}) and (\ref{eq:aromp}). This choice provides a long enough model order $p$ to be able to capture long-range dependencies in the data, as compared to asymptotic order selection criteria such as the Akaike or Bayesian information criteria, or the non-asymptotic minimax criterion of \cite{goldenshluger2001nonasymptotic}. In Section \ref{sec:ar_sim}, we will use this choice of the model order in our numerical and real data analysis.}

\noindent  \textit{\textbf{Remark 6.}} The main ingredient in the proofs of Theorems \ref{thm:ar_1}  and \ref{thm_OMP} is to establish the restricted eigenvalue (RE) condition introduced in \cite{bickel2009simultaneous} for the covariates matrix $\mathbf{X}$. Establishing the RE condition for the covariates matrix $\mathbf{X}$ is a nontrivial problem due to the high interdependence of the matrix entries. We will indeed show that if the sufficient stability assumption holds, then with $n \sim \mathcal{O}\left (s \max \{d_0 (\log p)^2, d_1 {(p \log p)^{1/2}} \} \right)$ the sample covariance matrix is sharply concentrated around the true covariance matrix and hence the RE condition can be guaranteed. All constants appearing in Theorems \ref{thm:ar_1} and \ref{thm_OMP} are explicitly given in Appendix \ref{app:ar_main}. As a typical numerical example, for $\eta = 0.9$ and $\sigma_w^2 = 0.1$, the constants of Theorem \ref{thm:ar_1} can be chosen as $d_0 \approx 1000, {d_1 \approx 3 \times 10^{8}}, {d_2 \approx  0.15} , d_3 \approx 140$, and  $d_4 = 1$. The full proofs are given in Appendix \ref{app:ar_main}.

\subsection{Minimax optimality}
\label{sec:ar_minimax}
{In this section}, we establish the minimax optimality of the $\ell_1$-regularized LS estimator for AR processes with sparse parameters. To this end, we will focus on the class $\mathcal{H}$ of stationary processes which admit an AR($p$) representation with $s$-sparse parameter $\boldsymbol{\theta}$ such that $\|\boldsymbol{\theta}\|_1 \leq 1-\eta <1$. The theoretical results of this section are inspired by the results of \cite{goldenshluger2001nonasymptotic} on non-asymptotic order selection {via $\ell_2$-regularized LS estimation in the absence of sparsity}, and extend them by studying the $\ell_1$-regularized LS estimator of (\ref{eq:ar_lasso}). 

We define the maximal \textit{estimation} risk over $\mathcal{H}$ to be
\begin{equation}
\label{eq:ar_minimax_risk}
\mathcal{R}_{\sf est} (\widehat{\boldsymbol{\theta}}):= \sup_{\mathcal{H}} \left(\mathbb{E}\left[\|\widehat{\boldsymbol{\theta}}-\boldsymbol{\theta}\|_2^2\right]\right)^{1/2}.
\end{equation}
The minimax estimator is the one minimizing the maximal estimation risk, i.e.,
\begin{equation}
\label{eq:ar_minimax_est}
\widehat{\boldsymbol{\theta}}_{\sf minimax} := \argmin\limits_{\boldsymbol{\theta}\in \boldsymbol{\Theta}} \quad \mathcal{R}_{\sf est} (\widehat{\boldsymbol{\theta}}).
\end{equation}
Minimax estimator $\widehat{\boldsymbol{\theta}}_{\sf minimax}$, in general, cannot be constructed explicitly \cite{goldenshluger2001nonasymptotic}, and the common practice in non-parametric estimation is to construct an estimator $\widehat{\boldsymbol{\theta}}$ which is \emph{order optimal} as compared to the minimax estimator:
\begin{equation}
\label{eq:ar_minimax_optinorder}
\mathcal{R}_{\sf est} (\widehat{\boldsymbol{\theta}}) \leq L \mathcal{R}_{\sf est} (\widehat{\boldsymbol{\theta}}_{\sf minimax}).
\end{equation}
with $L \ge 1$ being a constant. One can also define the minimax \textit{prediction} risk by the maximal prediction error over all possible realizations of the process:
\begin{equation}
\label{eq:ar_minimax_pred_risk}
\mathcal{R}^2_{\sf pre} (\widehat{\boldsymbol{\theta}}) := \sup_{\mathcal{H}} \mathbb{E}\left[\left(x_{k}-\widehat{\boldsymbol{\theta}}'\mathbf{x}_{k-p}^{k-1}\right)^2 \right].
\end{equation}
In \cite{goldenshluger2001nonasymptotic}, it is shown that an $\ell_2$-regularized LS estimator with an {order} $p^\star = \mathcal{O} (\log n) $ is minimax optimal. This order pertains to the denoising regime where $n \gg p$. Hence, in order to capture long order lags of the process, one requires a sample size exponentially large in $p$, which may make the estimation problem computationally infeasible. For instance, consider a $2$-sparse parameter with only $\theta_1$ and $\theta_p$ being non-zero. Then, in order to achieve minimax optimality, $n \sim \mathcal{O}(2^p)$ measurements are required. In contrast, in the compressive regime where $s, n \ll p$, the goal, instead of selecting $p$,  is to find conditions on the sparsity level $s$, so that for a given $n$ and large enough $p$, the $\ell_1$-regularized estimator is minimax optimal without explicit knowledge of the value of $s$ (See for example, \cite{candes2006modern}).

In the following proposition, we establish the minimax optimality of the $\ell_1$-regularized estimator over the class of sparse AR processes with $\boldsymbol{\theta} \in \boldsymbol{\Theta}$:
\begin{prop}
\label{thm:ar_minimax}
Let $\mathbf{x}_1^n$ be samples of an AR process with $s$-sparse parameters satisfying $\|\boldsymbol{\theta}\|_1 \leq 1-\eta$ and $s  \le  \min \left \{\frac{1- \eta}{\sqrt{8\pi} \eta} \sqrt{\frac{n}{\log p}}\ , \frac{n}{d_1{(p \log p)^{1/2}}}, {\frac{n}{d_0 (\log p)^2}} \right\}$. Then, we have:
\begin{equation*}
\mathcal{R}_{\sf est} (\widehat{\boldsymbol{\theta}}_{\ell_1}) \leq L \mathcal{R}_{\sf est} (\widehat{\boldsymbol{\theta}}_{\sf minimax}).
\end{equation*}
%\leq t_2 \mathcal{R}_{\sf est}(\widehat{\boldsymbol{\theta}}_{\sf minimax})
where $L$ is a constant and is only a function of $\eta$ and $\sigma_{\sf w}^2$ and is explicitly given in the proof.
%where $d_5$, $t_1$ and $t_2$ are only functions of $\eta$ and $\sigma_{\sf w}^2$ (modulo logarithmic factors in $p^{\star}$ for $t_2$) and are explicitly given in the proof.
\end{prop}

\noindent  \textit{\textbf{Remark 5.}} Proposition \ref{thm:ar_minimax} implies that $\ell_1$-regularized LS is minimax optimal in estimating the $s$-sparse parameter vector $\boldsymbol{\theta}$, for small enough $s$. The proof of the Proposition \ref{thm:ar_minimax} is given in Appendix \ref{prf:minimax}. This result can be extended to compressible $\boldsymbol{\theta}$ in a natural way with a bit more work, but we only present the proof for the case of $s$-sparse $\boldsymbol{\theta}$ for brevity. We also state the following proposition on the prediction performance of the $\ell_1$-regularized LS estimator:
\begin{prop}
\label{thm:ar_minimax2}
Let $\mathbf{x}_{-p+1}^n$ be samples of an AR process with $s$-sparse parameters and Gaussian innovations, then there exists  a positive constant $d_5$ such that for large enough $n,p$ and $s$ satisfying $n>d_1 s(p \log p)^{1/2}$, we have:
\begin{equation}
\mathcal{R}^2_{\sf pre}(\widehat{\boldsymbol{\theta}}_{\ell_1})  \leq d_5 \frac{s \log p}{n}+\sigma^2_{\sf w}.
\end{equation}
%\begin{equation}
%\mathcal{R}_p(\widehat{\boldsymbol{\theta}}_{\ell_1})  \leq c_{\eta}'\left(\frac{1}{n^2}+(1-\eta)^{2p}s+\frac{s}{n}\right)+\sigma_w^2.
%\end{equation}
\end{prop}
It can be readily observed that for {$n \gg s \log p$} the prediction error variance is very close to the variance of the innovations. The proof is similar to Theorem 3 of \cite{goldenshluger2001nonasymptotic} and is skipped in this paper for brevity.

%or equivalently if power spectral density is exactly zero over an interval or over more than one consecutive point.
%
%It can be shown that Mathematical arguments or simple reasoning will show that a function that is identically zero over some interval cannot be represented by a rational function of $e^{j\omega}$. Historically, this has been formulated as follows \cite{priestley1981spectral}: if
%\begin{equation}
%\int_{-\pi}^\pi \log S(\omega) d\omega > -\infty,
%\end{equation} then a unique one-sided $G(z)$ exists with the sequence $g_0, g_1, \cdots$ which has all zeros inside the unit circle. The integral of the logarithm will become $-\infty$ only if the power spectral density is exactly zero over an interval or over more than one consecutive point. The spectrum may touch zero at a single point, not at an interval. This is a mild requirement for spectra in practice. Therefore, almost all stationary stochastic processes can be modeled by a unique, stationary, and invertible ARMA process. In this paper we require a slightly stronger condition for the process to have a bounded spectral spread. 
\section{Application to Simulated and Real Data}
\label{sec:ar_sim}
In this section, we study and compare the performance of Yule-Walker based  estimation methods with those of the $\ell_1$-regularized and greedy estimators given in Section \ref{sec:ar_theory}. These methods are applied to simulated data as well as real data from crude oil price and traffic speed.

\subsection{Simulation studies}
In order to simulate an AR process, we filtered a Gaussian white noise process using an IIR filter with sparse parameters. Figure \ref{fig:sample_ar} shows a typical sample path of the simulated AR process used in our analysis. For the parameter vector  $\boldsymbol{\theta}$, we chose a length of $p=300$, and employed $n = 1500$ generated samples of the corresponding process for estimation. The parameter vector $\boldsymbol{\theta}$ is of sparsity level $s=3$ and $\eta = 1-\|\boldsymbol{\theta}\|_1=0.5$. A value of $\gamma_n = 0.1$ is used, which is slightly tuned around the theoretical estimate given by Theorem \ref{thm:ar_1}. The order of the process is assumed to be known. We compare the performance of seven estimators: 1) $\widehat{\boldsymbol{\theta}}_{\sf LS}$ using LS, 2) $\widehat{\boldsymbol{\theta}}_{\sf yw}$ using the Yule-Walker equations, 3) $\widehat{\boldsymbol{\theta}}_{\ell_1}$ from $\ell_1$-regularized LS, 4) $\widehat{\boldsymbol{\theta}}_{\sf OMP}$ using OMP,  5) $\widehat{\boldsymbol{\theta}}_{\sf yw, \ell_{2,1}}$ using Eq. (\ref{eq:ar_ywl21}), 6) $\widehat{\boldsymbol{\theta}}_{\sf yw, \ell_{1,1}}$ using Eq. (\ref{eq:ar_ywl11}), and 7) $\widehat{\boldsymbol{\theta}}_{\sf ywOMP}$ using the cost function $\mathfrak{J}(\boldsymbol{\theta})$ in the generalized OMP. {Note that for the LS and Yule-Walker estimates, we have relaxed the condition of $\| \boldsymbol{\theta} \|_1 < 1$, to be consistent with the common usage of these methods. The Yule-Walker estimate is guaranteed to result in a stable AR process, whereas the LS estimate is not \cite{percival1993spectral}.} Figure \ref{fig:ar_param} shows the estimated parameter vectors using these algorithms. It can be visually observed that $\ell_1$-regularized and greedy estimators (shown in {purple}) significantly outperform the Yule-Walker-based estimates (shown in {orange}).

\begin{figure}[h]
\begin{center}
\noindent
\includegraphics[width=.9\columnwidth]{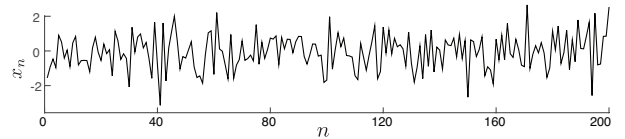}
\caption{Samples of the simulated AR process.}\label{fig:sample_ar}
\end{center}
\vspace{-5mm}
\end{figure}

\begin{figure}[h]
\begin{center}
\noindent
\includegraphics[width=.9\columnwidth]{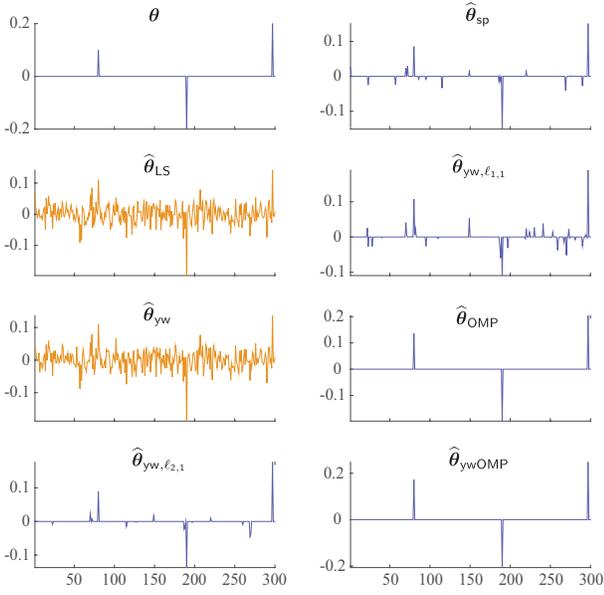}
\caption{Estimates of $\boldsymbol{\theta}$ for $n=1500$, $p =300$, and $s=3$ (These results are best viewed in the color version).}\label{fig:ar_param}
\end{center}
\vspace{-5mm}
\end{figure}

In order to quantify the latter observation precisely, we repeated the same experiment for $p=300, s=3$ and $10 \leq n \leq 10^5$. A comparison of the normalized MSE of the estimators vs. $n$ is shown in Figure \ref{fig:ar_mse}. As it can be inferred from Figure \ref{fig:ar_mse}, in the region where $n$ is comparable to or less than $p$ {(shaded in light {purple})}, the sparse estimators have a systematic performance gain over the Yule-Walker based estimates, with the $\ell_1$-regularized LS and ywOMP estimates outperforming the rest.

\begin{figure}[h]
\begin{center}
\noindent
\includegraphics[width=.85\columnwidth]{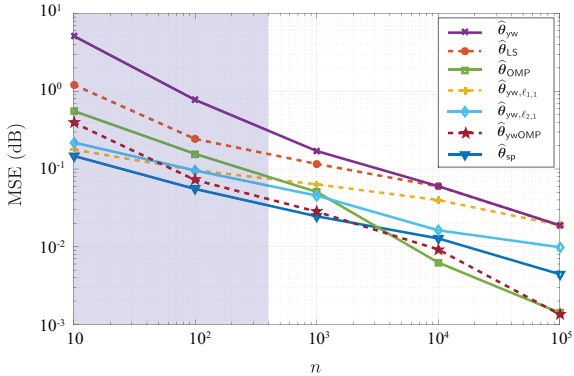}
\caption{MSE comparison of the estimators vs. the number of measurements $n$. {The shaded region corresponds to the compressive regime of $n < p$.}}\label{fig:ar_mse}
\end{center}
\vspace{-5mm}
\end{figure}

The MSE comparison in Figure \ref{fig:ar_mse} requires one to know the true parameters. In practice, the true parameters are not available for comparison purposes. In order to quantify the performance gain of these methods, we use statistical tests to assess the goodness-of-fit of the estimates. The common chi-square type statistical tests, such as the F-test, are useful when the hypothesized distribution to be tested against is discrete or categorical. For our problem setup with sub-Gaussian innovations, we will use a number of statistical tests appropriate for AR processes, namely, the Kolmogorov-Smirnov (KS) test, the Cram\'er-von Mises (CvM) criterion, {the spectral Cram\'er-von Mises (SCvM) test} and the Anderson-Darling (AD) \cite{d1986goodness,johansen1995likelihood,anderson1997goodness}. A summary of these tests is given in Appendix \ref{app:ar_tests}. Table \ref{tab:synthetic_table} summarizes the test statistics for different estimation methods. {Cells colored in orange (darker shade in grayscale) correspond to traditional AR estimation methods and those colored in blue (lighter shade in grayscale) correspond to the sparse estimator with the best performance among those considered in this work}. These tests are based on the known results on limiting distributions of error residuals. As noted from Table \ref{tab:synthetic_table}, our simulations suggest that the OMP estimate achieves the best test statistics for the CvM, AD and KS tests, whereas the $\ell_1$-regularized estimate achieves the best SCvM statistic.

% \textbf{For comparison purposes, after the model is fitted, differencing is converted and all forcast values are constructed for the oroginal series.}

\begin{table}[h!]
\centering
\caption{Goodness-of-fit tests for the simulated data}
\label{tab:synthetic_table}
\begin{tabular}{lllll}
\multicolumn{1}{l|}{\backslashbox{Estimate}{Test}} & CvM                  & AD                   & KS                  & SCvM \\ \cline{1-5} \hline
\multicolumn{1}{l|}{${\boldsymbol{\theta}}$}   &  0.31          &  1.54  &  0.031    &   0.009      \\
\multicolumn{1}{l|}{$\widehat{\boldsymbol{\theta}}_{\sf LS}$}  &   \cellcolor{orange!70}0.68         &   \cellcolor{orange!70}5.12         &  \cellcolor{orange!70}0.037      &  \cellcolor{orange!70}0.017   \\
\multicolumn{1}{l|}{$\widehat{\boldsymbol{\theta}}_{\sf yw}$}  & \cellcolor{orange!70}0.65 & \cellcolor{orange!70}4.87 &  \cellcolor{orange!70}0.034 &  \cellcolor{orange!70}0.025 \\
\multicolumn{1}{l|}{$\widehat{\boldsymbol{\theta}}_{\ell_1}$}  & 0.34 & 1.72 &  0.030  &    \cellcolor{blue!10}0.009        \\
\multicolumn{1}{l|}{$\widehat{\boldsymbol{\theta}}_{\sf OMP}$}  & \cellcolor{blue!10}0.29 &            \cellcolor{blue!10}1.45  &   \cellcolor{blue!10}0.028     &   0.009    \\
\multicolumn{1}{l|}{$\widehat{\boldsymbol{\theta}}_{{\sf yw},\ell_{2,1}}$}  &  0.35 & 1.80 &   0.032   & 0.009  \\
\multicolumn{1}{l|}{$\widehat{\boldsymbol{\theta}}_{{\sf yw},\ell_{1,1}}$}  &         0.42  &   2.33       &    0.040     &   0.008  \\
\multicolumn{1}{l|}{$\widehat{\boldsymbol{\theta}}_{\sf ywOMP}$}  &  0.29 &  1.46 &        0.030 &   0.009
\end{tabular}
\vspace{-2mm}
\end{table}

%\color{blue}
%\noindent \textbf{\textit{ 7:}} The asymptotic model order selection criteria of \cite{goldenshluger2001nonasymptotic} results in an underestimated order for the process, and therefore the small choice of mode order is not capable of capturing the sparse and long order structure of the parameters. Throughout this section, we have instead used the order selection described in Remark 4.
%

\subsection{Application to the analysis of \textcolor{black}{crude oil prices}}

In this and the following subsection, we consider applications with real-world data. As for the first application, we apply the {sparse} AR estimation techniques to analyze the crude oil price of the Cushing, OK WTI Spot Price FOB dataset \cite{cushing}. This dataset consists of 7429 daily values of oil prices in dollars per barrel. In order to avoid outliers, usually the dataset is filtered with a moving average filter of high order. We have skipped this procedure by visual inspection of the data and selecting $n=4000$ samples free of outliers. Such financial data sets are known for their non-stationarity and long order history dependence. In order to remove the deterministic trends in the data, one-step or two-step time differencing is typically used. We refer to \cite{robinson2003time} for a full discussion of this detrending method. We have used a first-order time differencing which resulted in a sufficient detrending of the data. Figure \ref{fig:sample_oil} shows the data used in our analysis. We have chosen $p=150$ by inspection. The histogram of first-order differences as well the estimates are shown in Figure \ref{fig:oil_estimate}.

\begin{figure}[h]
\begin{center}
\noindent
\includegraphics[width=.9\columnwidth]{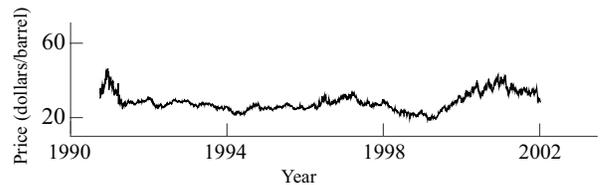}
\caption{{A sample segment of the Cushing, OK WTI Spot Price FOB data.}}\label{fig:sample_oil}
\end{center}
\vspace{-2mm}
\end{figure}

\begin{figure}[h!]
\begin{center}
\noindent
\includegraphics[width=.9\columnwidth]{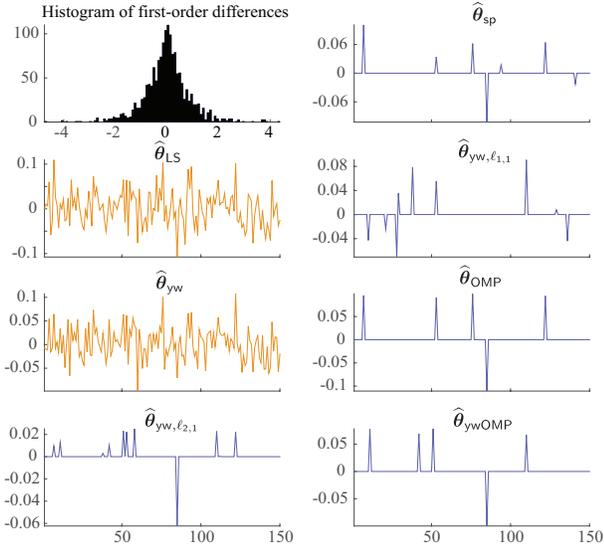}
\caption{Estimates of $\boldsymbol{\theta}$ for the second-order differences of the oil price data.}\label{fig:oil_estimate}
\end{center}
\vspace{-2mm}
\end{figure}

A visual inspection of the estimates in Figure \ref{fig:oil_estimate} shows that the $\ell_1$-regularized LS ($\widehat{\boldsymbol{\theta}}_{\ell_1}$) and OMP ($\widehat{\boldsymbol{\theta}}_{\sf OMP}$)  estimates consistently select specific time lags in the AR parameters, whereas the Yule-Walker and LS estimates seemingly overfit the data by populating the entire parameter space. In order to perform goodness-of-fit tests, we use an even/odd two-fold cross-validation. Table \ref{tab:oil_table} shows the corresponding test statistics, which reveal that indeed the $\ell_1$-regularized and OMP estimates outperform the traditional estimation techniques.

\begin{table}[h]
\centering
\caption{{Goodness-of-fit tests for the crude oil price data}}
\label{tab:oil_table}
\begin{tabular}{lllll}
\multicolumn{1}{l|}{\backslashbox{Estimate}{Test}} & CvM                  & AD                   & KS                  & SCvM \\ \cline{1-5} \hline
\multicolumn{1}{l|}{$\widehat{\boldsymbol{\theta}}_{\sf LS}$}  &   \cellcolor{orange!70}0.88         &   \cellcolor{orange!70}5.55        &  \cellcolor{orange!70}0.055      &  \cellcolor{orange!70}0.046   \\
\multicolumn{1}{l|}{$\widehat{\boldsymbol{\theta}}_{\sf yw}$}  & \cellcolor{orange!70}0.58 & \cellcolor{orange!70}3.60 &  \cellcolor{orange!70}0.043 &  \cellcolor{orange!70}0.037 \\
\multicolumn{1}{l|}{$\widehat{\boldsymbol{\theta}}_{\ell_1}$}  & 0.27 & 1.33 &  0.031  &    \cellcolor{blue!10}0.020       \\
\multicolumn{1}{l|}{$\widehat{\boldsymbol{\theta}}_{\sf OMP}$}  & \cellcolor{blue!10}0.22 &            \cellcolor{blue!10}1.18  &   \cellcolor{blue!10}0.025     &   0.022    \\
\multicolumn{1}{l|}{$\widehat{\boldsymbol{\theta}}_{{\sf yw},\ell_{2,1}}$}  &  0.28 & 1.40 &   0.027   & 0.021  \\
\multicolumn{1}{l|}{$\widehat{\boldsymbol{\theta}}_{{\sf yw},\ell_{1,1}}$}  &         0.24  &   1.26       &    0.027     &   0.022  \\
\multicolumn{1}{l|}{$\widehat{\boldsymbol{\theta}}_{\sf ywOMP}$}  &  0.23 &  \cellcolor{blue!10} 1.18 &        0.026 &   0.022
\end{tabular}
\vspace{-2mm}
\end{table}

\subsection{Application to the analysis of traffic data}
Our second real data application concerns traffic speed data. The data used in our simulations is the INRIX \textregistered\  speed data for I-495 Maryland inner loop freeway (clockwise) between US-1/Baltimore Ave/Exit 25 and Greenbelt Metro Dr/Exit 24 from 1 Jul, 2015 to 31 Oct, 2015 \cite{ritis1,ritis2}. The reference speed of 65 mph is reported. {Our aim is to analyze the long-term, large-scale periodicities manifested in these data by fitting high-order sparse AR models}. Given the huge length of the data and its high variability, the following pre-processing was made on the original data:
\begin{enumerate}
\item The data was downsampled by a factor of $4$ and averaged by the hour in order to reduce its daily variability, that is each lag corresponds to one hour.
\item The logarithm of speed was used for analysis and the mean was subtracted.  This reduces the high variability of speed due to rush hours and lower traffic during weekends and holidays.
\end{enumerate}

\begin{figure}[h]
\begin{center}
\noindent
\includegraphics[width=0.8\columnwidth]{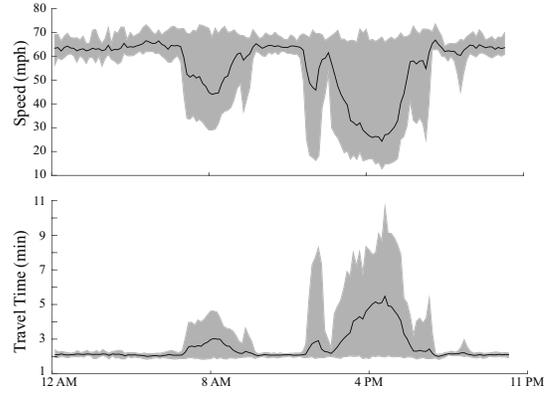}
\caption{{A sample of the speed and travel time data for I-495.}}\label{fig:ar_speed}
\end{center}
\vspace{-2mm}
\end{figure}

Figure \ref{fig:ar_speed} shows a typical average weekly speed and travel time in this dataset and the corresponding 25-75-th percentiles. As can be seen the data shows high variability around the rush hours of $8~\text{am}$ and $4~\text{pm}$. In our analysis, we used the first half of the data ($n=1500$) for fitting, from which the AR parameters and the distribution and variance of the innovations were estimated. The statistical tests were designed based on the estimated distributions, and the statistics were computed accordingly using the second half of the data. \textcolor{black}{We selected an order of $p = 200$ by inspection and noting that the data seems to have a periodicity of order $170$ samples.}

\begin{figure}[h]
\begin{center}
\noindent
\includegraphics[width=.9\columnwidth]{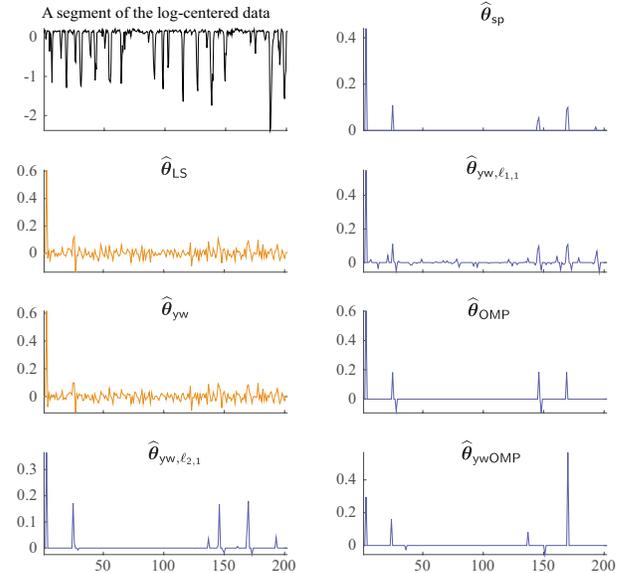}
\caption{{Estimates of $\boldsymbol{\theta}$ for the traffic speed data.}}\label{fig:traffic_estimate}
\end{center}
\vspace{-2mm}
\end{figure}

Figure \ref{fig:traffic_estimate} shows part of the data used in our analysis as well as the estimated parameters. The $\ell_1$-regularized LS ($\widehat{\boldsymbol{\theta}}_{\ell_1}$) and OMP ($\widehat{\boldsymbol{\theta}}_{\sf OMP}$) are consistent in selecting the same components of $\boldsymbol{\theta}$. These estimators pick up two major lags around which $\boldsymbol{\theta}$ has its largest components. The first lag corresponds to about $24$ hours which is mainly due to the rush hour periodicity on a daily basis. The second lag is around $150-170$ hours which corresponds to weekly changes in the speed due to lower traffic over the weekend. In contrast, the Yule-Walker and LS estimates do not recover these significant time lags. 

\begin{table}[h]
\centering
\vspace{-2mm}
\caption{Goodness-of-fit tests for the traffic speed data}
\label{tab:traffic_table}
\begin{tabular}{lllll}
\multicolumn{1}{l|}{\backslashbox{Estimate}{Test}} & CvM                  & AD                   & KS                  & SCvM \\ \cline{1-5} \hline
\multicolumn{1}{l|}{$\widehat{\boldsymbol{\theta}}_{\sf yw}$}  & \cellcolor{orange!70}0.012 & \cellcolor{orange!70}0.066 &  \cellcolor{orange!70}0.220 &  \cellcolor{orange!70} 0.05 \\
\multicolumn{1}{l|}{$\widehat{\boldsymbol{\theta}}_{\ell_1}$}  &  \cellcolor{blue!10}1.4$\times 10^{-7}$ &  \cellcolor{blue!10}2.1$\times 10^{-6}$ &   \cellcolor{blue!10}6.7$\times 10^{-4}$  &   0.25       \\
\multicolumn{1}{l|}{$\widehat{\boldsymbol{\theta}}_{\sf OMP}$}  & 0.017 &  0.082  & 0.220 &  1.49    \\
\multicolumn{1}{l|}{$\widehat{\boldsymbol{\theta}}_{\sf ywOMP}$}  &  0.025 &  0.122 &        0.270 &   0.14
\end{tabular}
\end{table}

Statistical tests for a selected subset of the estimators are shown in Table \ref{tab:traffic_table}. Interestingly, the $\ell_1$-regularized LS estimator significantly outperforms the other estimators in three of the tests. The Yule-Walker estimator, however, achieves the best SCvM test statistic.

\section{Conclusions}\label{sec:conc}
In this paper, we {have} investigated sufficient sampling requirements for stable estimation of AR models in the non-asymptotic regime using the $\ell_1$-regularized LS and greedy estimation (OMP) techniques. We {have} further established the minimax optimality of the $\ell_1$-regularized LS estimator. Compared to the existing literature, our results provide several major contributions. {First, when $s \sim p^{\frac{1}{2} + \delta}$ for some $\delta \ge 0$, our results suggest an improvement of order $\mathcal{O}(p^{\delta} (\log p)^{3/2})$ in the sampling requirements for the estimation of univariate AR models with sub-Gaussian innovations using the LASSO, over those of \cite{wong2016regularized} and \cite{wu2016performance} which require $n \sim \mathcal{O}(p^2 (\log p)^2)$ for stable AR estimation.} {When specialized to a sub-Gaussian white noise process, i.e., establishing the RE condition of i.i.d. Toeplitz matrices, our results provide an improvement of order $\mathcal{O}(s/\log p)$ over those of \cite{Toeplitz}. Second, although OMP is widely used in practice, the choice of the number of greedy iterations is often ad-hoc. In contrast, our theoretical results prescribe an analytical choices of the number of iterations required for stable estimation, thereby promoting the usage of OMP as a low-complexity algorithm for AR estimation. Third, we established the minimax optimality of the $\ell_1$-regularized LS estimator for the estimation of sparse AR parameters.

We {further} verified the validity of our theoretical results through simulation studies as well as application to real financial and traffic data. These results show that the sparse estimation methods significantly outperform the widely-used Yule-Walker based estimators in fitting AR models to the data. Although we did not theoretically analyze the performance of sparse Yule-Walker based estimators, they seem to perform on par with the $\ell_1$-regularized LS and OMP estimators based on our numerical studies. Finally, our results provide a striking connection to our recent work \cite{kazemipour, kazemipour2015robust} in estimating sparse self-exciting discrete point process models. These models regress an observed binary spike train with respect to its history via Bernoulli or Poisson statistics, and are often used in describing spontaneous activity of sensory neurons. Our results {have shown} that in order to estimate a sparse history-dependence parameter vector of length $p$ and sparsity $s$ in a stable fashion, a spike train of length $n \sim \mathcal{O}(s^{2/3}p^{2/3}\log p)$ is required. {This leads us to conjecture} that these sub-linear sampling requirements are sufficient for a larger class of autoregressive processes, beyond those characterized by linear models.} {Finally, our minimax optimality result requires the sparsity level $s$ to grow at most as fast as $\mathcal{O}(n/(p\log p)^{1/2})$. We consider further relaxation of this condition, as well as the generalization of our results to sparse MVAR processes as future work.}

%{Moreover our results can be generalized to Multivariate AR models via a multivariate version of the concentration inequalities of  \cite{rudzkis1978large} in connjunction with the techniques of \cite{ wong2016regularized} and \cite{wu2016performance}}.
\appendices

\section{Proofs of Theorems \ref{thm:ar_1} and \ref{thm_OMP}}

\subsection{The Restricted Strong Convexity of the matrix of covariates}
The first element of the proofs of both Theorems \ref{thm:ar_1} and \ref{thm_OMP} is to establish the Restricted Strong Convexity (RSC) for the matrix $\mathbf{X}$ of covariates formed from the observed data. First, we investigate the closely related {Restricted} Eigenvalue (RE) condition. Let $[\lambda_{\sf min}(s)$, $\lambda_{\sf max}(s)]$ be the smallest interval containing the singular values of $\frac{1}{n} (\mathbf{X}_S^T \mathbf{X}_S)$, where $\mathbf{X}_S$ is a sub-matrix $\mathbf{X}$ over an index set $S$ of size $s$.

\begin{definition}[Restricted Eigenvalue Condition]
\label{RE_def}
A matrix $\mathbf{X}$ is said to satisfy the RE condition of order $s$ if $\lambda_{\sf min}(s) > 0$.
\end{definition}
Although the RE condition only restricts $\lambda_{\sf min}(s)$, in the following analysis we also keep track of $\lambda_{\sf max}(s)$, which appears in some of the bounds. Establishing the RSC for $\mathbf{X}$ proceeds in a sequence of lemmas (Lemmas \ref{eig_conv}--\ref{RE_RSC} culminating in Lemma \ref{lem:rsc}). We first show that the RE condition holds for the true covariance of an AR process:
\begin{lemma}[from \cite{grenander1958toeplitz}]
\label{eig_conv}
Let $\mathbf{R} \in \mathbb{R}^{k \times k}$ be the $k \times k$ covariance matrix of a stationary process with power spectral density $S(\omega)$, and denote its maximum and minimum eigenvalues by $\phi_{\max}(k)$ and $\phi_{\sf min}(k)$, respectively. Then, $\phi_{\max}(k)$ is increasing in $k$, $\phi_{ \sf min}(k)$ is decreasing in $k$, and we have
\begin{equation}
\phi_{\sf min}(k) \downarrow \inf_{\omega}S(\omega), \quad \mbox{and} \quad \phi_{\sf max}(k) \uparrow \sup_{\omega}S(\omega).
\end{equation}
\end{lemma}
\noindent This result gives us the following corollary:
\begin{corollary}[Singular Value Spread of $\mathbf{R}$]
\label{cor:eig_conv}
Under the {sufficient stability assumption}, the singular values  of the covariance $\mathbf{R}$ of an AR process lie in the interval $\left [\frac{\sigma^2_{\sf w}}{8 \pi}, \frac{\sigma^2_{\sf w}}{2 \pi \eta^2} \right]$.
\end{corollary}
\begin{proof}
For an AR($p$) process
\[
S(\omega) = \frac{1}{2\pi}\frac{\sigma^2_{\sf w}}{|1- \sum_{\ell=1}^p \theta_{\ell} e^{-j\ell\omega}|^2}.
\]
Combining $\|\boldsymbol{\theta}\|_1 \leq 1-\eta < 1$ with Lemma \ref{eig_conv} proves the claim.
\end{proof}

Note that by Lemma \ref{eig_conv}, the result of Corollary \ref{eig_conv} not only holds for AR processes, but also for \textit{any} stationary process satisfying $\inf_\omega S(\omega) >0$ and $\sup_\omega S(\omega) < \infty$, i.e., a process with finite spectral spread.

We next establish conditions for the RE condition to hold for the empirical covariance $\widehat{\mathbf{R}}$:

\begin{lemma}\label{lem:re}
If the singular values of $\mathbf{R}$ lie in the interval $[\lambda_{\sf min}, \lambda_{\sf max}]$, then $\mathbf{X}$ satisfies the RE condition of order { $s_\star$} with parameters ${\lambda}_{\sf min}(s_\star) = \lambda_{\sf min} - t s_\star$ and ${\lambda}_{\sf max} (s_\star)= \lambda_{\sf max} +ts_\star$, where $t = \max_{i,j} |\widehat{R}_{ij}-R_{ij}|$.
\end{lemma}
\begin{proof}
Let $\widehat{\mathbf{R}} = \frac{1}{n} (\mathbf{X}^T \mathbf{X})$. For every $s_\star$-sparse $\boldsymbol{\theta}$ we have
\[
\boldsymbol{\theta}^T \widehat{\mathbf{R}} \boldsymbol{\theta} \geq \boldsymbol{\theta}^T {\mathbf{R}} \boldsymbol{\theta} - t \|\boldsymbol{\theta}\|_1^2 \geq (\lambda_{\sf min} - t s_\star) \|\boldsymbol{\theta}\|_2^2,
\]
\[\boldsymbol{\theta}^T \widehat{\mathbf{R}} \boldsymbol{\theta} \leq \boldsymbol{\theta}^T {\mathbf{R}} \boldsymbol{\theta} + t \|\boldsymbol{\theta}\|_1^2 \leq (\lambda_{\sf max} + t s_\star) \|\boldsymbol{\theta}\|_2^2,\]
which proves the claim.
\end{proof}

We will next show that $t$ can be suitably controlled with high probability. Before doing so, we state a {key} result of Rudzkis \cite{rudzkis1978large} regarding the concentration of second-order empirical sums from stationary processes:
\begin{lemma}
\label{conc_biased}
Let $\mathbf{x}_{-p+1}^n$ be samples of a stationary process which satisfies 
\vspace{-.2cm}
\begin{equation}
\label{eq:wold}
x_k = \sum_{j= -\infty}^\infty b_{j-k} w_j,
\vspace{-.2cm}
\end{equation}
\noindent where $w_k$'s are i.i.d random variables with
\begin{equation}
\label{eq: bounded_moments}
|\mathbb{E}(|w_j|^k)| \leq ({\tilde{c} \sigma_{\sf w}})^k k!, \ k=2, 3, \cdots,
\end{equation} 
for some constant $\tilde{c}$ and
\vspace{-.2cm}
\begin{equation}
\label{eq: abs_sum}
\sum_{j=-\infty}^\infty |b_j| < \infty.
\end{equation}
\vspace{-.2cm}
Then, the \textit{biased} sample autocorrelation given by
\[\widehat{r}^b_k=\frac{1}{n+k}\sum_{i,j=1,j-i=k}^{n+k}x_ix_j\]
\vspace{-.2cm}
\noindent satisfies
\begin{equation}
\label{conc_ineq_biased}
\mathbb{P}(|\widehat{r}^b_k - r^b_k|>t) \leq c_1 (n+k) \exp \left(-{\frac{c_2}{\sigma_{\sf w}} \frac{t^2 (n+k)}{c_3 \sigma_{\sf w}^3 + t^{3/2} \sqrt{n+k}}}\right),
\end{equation}
for { positive absolute constants $c_1$, $c_2$ and $c_3$ which are independent of the dimensions of the problem. In particular, if $x_k = w_k$, i.e., a sub-Gaussian white noise process, $c_3$ vanishes.}
\end{lemma}
\begin{proof}
The lemma is a special case of Theorem 4 under Condition 2 of Remark 3 in \cite{rudzkis1978large}. {For the special case of $x_k = w_k$, the constant $H$ in Lemma 7 of \cite{rudzkis1978large} and hence $c_3$ vanish.}
\end{proof}

{Using the result of} Lemma \ref{conc_biased}, we can control $t$ and establish the RE condition for $\widehat{\mathbf{R}}$ as follows:
\begin{lemma}\label{lem:re2}
Let $m$ be a positive integer. Then, $\mathbf{X}$ satisfies the RE condition of order $(m+1)s$ with a constant $\lambda_{\sf min}/2$ with probability at least
\vspace{-.4cm}
\begin{equation}
1 - c_1 p^2 (n+p) \exp \left(-\frac{c_4 \sqrt{\frac{n}{s}}}{1 + c_5 \frac{n+p}{\left(\frac{n}{s}\right)^{3/2}}}\right),
\end{equation}
\vspace{-.2cm}
\noindent where $c_1$ is the same as in Lemma \ref{conc_biased}, $c_4 = \frac{c_2}{\sigma_{\sf w}} \sqrt{ \frac{\lambda_{\sf min}}{2(m+1)}}$ and $c_5 = \frac{c_3 \sigma_{\sf w}^3}{\left( \frac{\lambda_{\sf min}}{2(m+1)}\right)^{3/2}}$.
\end{lemma}

\begin{proof}
First, note that for the given AR process, condition (\ref{eq:wold}) is verified by the Wold decomposition of the process, condition (\ref{eq: bounded_moments}) results from the sub-Gaussian assumption on the innovations, and condition (\ref{eq: abs_sum}) results from the stability of the process. Noting that
\vspace{-.2cm}
\begin{equation}
\widehat{R}_{i,i+k} = \frac{1}{n}\sum_{i=1}^n x_i x_{i+k} = \frac{1}{n}\sum_{i,j=1,j-i=k}^{n+k}x_ix_j = \frac{n+k}{n}\widehat{r}^b_k, 
\end{equation}
for $i=1,\cdots,n$ and $k = 0, \cdots, p-1$, Eq. (\ref{conc_ineq_biased}) implies:
\begin{equation}
\nonumber \resizebox{\columnwidth}{!}{$\displaystyle \mathbb{P}\left ( |\widehat{R}_{i,i+k} - {R}_{i,i+k}|> \tau \right) \leq c_1 (n+k) \exp \left( - \frac{ c_2 \sqrt{\tau n}}{\frac{c_3 \sigma_{\sf w}^4 (n+k)}{\tau^{3/2} n^{3/2}} + \sigma_{\sf w}}\right)$}.
\end{equation}
By the union bound and $k \le p$, we get:
\begin{align}\label{eq:max}
\resizebox{\columnwidth}{!}{$\displaystyle \mathbb{P}\left(\max_{i,j}|\widehat{R}_{ij}-R_{ij}|> \tau \right) \leq c_1 p^2 (n+p) \exp \left( - \frac{ c_2 \sqrt{\tau n}}{\frac{c_3 \sigma_{\sf w}^4 (n+p)}{\tau^{3/2} n^{3/2}} + \sigma_{\sf w}}\right)$}.
\end{align}
Choosing $\tau= \frac{\lambda_{\sf min}}{2(m+1)s}$ and invoking the result of Lemma \ref{lem:re} establishes the result of the lemma.
\end{proof}

We next define the closely related notion of the Restricted Strong Convexity (RSC):

\begin{definition}[Restricted Strong Convexity \cite{Negahban}]
\label{RSC_def}
Let 
\begin{equation}
\label{cone_condition}
\mathbb{V}:= \{\mathbf{h}\in \mathbb{R}^p | \|\mathbf{h}_{S^c}\|_1 \leq 3\| \mathbf{h}_S\|_1+4\|\boldsymbol{\theta}_{S^c}\|_1\}.
\end{equation}
Then, $\mathbf{X}$ is said to satisfy the RSC condition of order $s$ if there exists a positive $\kappa > 0$ such that
\begin{equation}
\frac{1}{n}\mathbf{h}^T \mathbf{X}^T \mathbf{X} \mathbf{h} = \frac{1}{n}\|\mathbf{X}\mathbf{h}\|_2^2 \geq \kappa \|\mathbf{h}\|_2^2, \;\;\;\; \forall \mathbf{h} \in \mathbb{V}.
\end{equation}
\end{definition}

The RSC condition can be deduced from the RE condition according to the following result:
\begin{lemma}[Lemma 4.1 of \cite{bickel2009simultaneous}]
\label{RE_RSC}
If $\mathbf{X}$ satisfies the RE condition of order $s_\star = (m+1)s$ with a constant $\lambda_{\sf min}((m+1)s)$, then the RSC condition of order $s$ holds with
\vspace{-.2cm}
\begin{equation}
\kappa = {\lambda_{\sf min}((m+1)s)}\left( 1- 3 \sqrt{\frac{\lambda_{\sf max}(ms)}{m\lambda_{\sf min}\left((m+1)s\right)}}\right)^2.
\end{equation}
\end{lemma}

We can now establish the RSC condition of order $s$ for $\mathbf{X}$:
\begin{lemma}\label{lem:rsc}
The matrix of covariates $\mathbf{X}$ satisfies the RSC condition of order $s$ with a constant $\kappa = \frac{\sigma^2_{\sf w}}{16 \pi}$ with probability at least
\vspace{-.3cm}
\begin{equation}
\label{ar:eq_rsc}
1 - c_1 p^2 (n+p) \exp \left(-\frac{c_{\eta} \sqrt{\frac{n}{s}}}{1 + c'_{\eta} \frac{n+p}{\left(\frac{n}{s}\right)^{3/2}}}\right),\end{equation}
where $c_\eta = \frac{c_2 \eta}{\sqrt{16 \pi ( 72 + \eta^2)}}$ and $c'_\eta = \frac{c_3 (16 \pi (72 + \eta^2))^{3/2}}{\eta^3}$.
\end{lemma}
\begin{proof}
Choosing $m = \lceil \frac{72}{\eta^2} \rceil$, and using Lemmas \ref{lem:re}, \ref{lem:re2}, and \ref{RE_RSC} establishes the result. {Note that if $x_k = w_k$, i.e., a sub-Gaussian white noise process, then $c_3$ and hence $c'_\eta$ vanish.}
\end{proof}

We are now ready prove Theorems \ref{thm:ar_1} and \ref{thm_OMP}.

\subsection{Proof of Theorem \ref{thm:ar_1}}

We first establish the so-called vase (cone) condition for the error vector $\mathbf{h} = \widehat{\boldsymbol{\theta}}_{\ell_1}-{\boldsymbol{\theta}}$:

\label{app:ar_main}
\begin{lemma}
For a choice of the regularization parameter $\gamma_n \ge \| \nabla \mathfrak{L}(\boldsymbol{\theta}) \|_\infty = \frac{2}{n} \| \mathbf{X}^T \left(\mathbf{x}_1^n-\mathbf{X}\boldsymbol{\theta}\right) \|_{\infty}$, the optimal error $\mathbf{h} = \widehat{\boldsymbol{\theta}}_{\ell_1}-{\boldsymbol{\theta}}$ belongs to the vase
\begin{equation}
\mathbb{V}:= \{\mathbf{h}\in \mathbb{R}^p | \|\mathbf{h}_{S^c}\|_1 \leq 3\| \mathbf{h}_S\|_1+4\|\boldsymbol{\theta}_{S^c}\|_1\}.
\end{equation}
\end{lemma}

\begin{proof}
Using several instances of the triangle inequality we have:
\begin{align*}
0 & \geq  \frac{1}{n} \left(\|\mathbf{x}_1^n-\mathbf{X}(\boldsymbol{\theta}+\mathbf{h})\|_2^2- \|\mathbf{x}_1^n-\mathbf{X}\boldsymbol{\theta}\|_2^2 \right)+ \\
& \;\;\;\;\; \gamma_n \left( \|\boldsymbol{\theta}+\mathbf{h}\|_1 - \|\boldsymbol{\theta}\|_1 \right)\\
& \geq -  \frac{1}{n} \| \mathbf{X}^T \left(\mathbf{x}_1^n-\mathbf{X}\boldsymbol{\theta}\right) \|_{\infty} \|\mathbf{h}\|_1 +\\
& \;\;\;\;\; \gamma_n \left( \|\boldsymbol{\theta}_S+\mathbf{h}_{S^c}+\mathbf{h}_S +\boldsymbol{\theta}_{S^c}\|_1 - \|\boldsymbol{\theta}\|_1 \right)\\
& \geq -  \frac{\gamma_n}{2} (\|\mathbf{h}_{S^c}\|_1+\|\mathbf{h}_S\|_1) +\\
& \;\;\;\;\; \gamma_n \left( \|\boldsymbol{\theta}_S+\mathbf{h}_{S^c}\|_1-\|\mathbf{h}_S +\boldsymbol{\theta}_{S^c}\|_1 - \|\boldsymbol{\theta}\|_1 \right)\\
& = -  \frac{\gamma_n}{2} (\|\mathbf{h}_{S^c}\|_1+\|\mathbf{h}_S\|_1) + \\
& \;\;\;\;\;\gamma_n (\|\boldsymbol{\theta}_S\|_1+\|\mathbf{h}_{S^c}\|_1-\|\mathbf{h}_S\|_1-\|\boldsymbol{\theta}_{S^c}\|_1-\|\boldsymbol{\theta}_{S^c}\|_1- \|\boldsymbol{\theta}_S\|_1)\\
& =  \frac{\gamma_n}{2} (\|\mathbf{h}_{S^c}\|_1-3\|\mathbf{h}_{S}\|_1 - 4\|\boldsymbol{\theta}_{S^c}\|_1).
\end{align*}
\end{proof}

The following result of Negahban et al. \cite{Negahban} allows us to characterize the desired error bound: 
\begin{lemma}[Theorem 1 of \cite{Negahban}]
\label{RSC_thm}
If $\mathbf{X}$ satisfies the RSC condition of order $s$ with  a constant $\kappa > 0$ and $\gamma_n \ge \| \nabla \mathfrak{L}(\boldsymbol{\theta}) \|_\infty$, then any optimal solution $\widehat{\boldsymbol{\theta}}_{\ell_1}$ satisfies
\begin{align}\label{eq:rsc_bound}
\nonumber & \|\widehat{\boldsymbol{\theta}}_{\ell_1}-\boldsymbol{\theta}\|_2 \leq  \frac{2\sqrt{s}\gamma_n}{\kappa}+ \sqrt{\frac{2\gamma_n \sigma_s(\boldsymbol{\theta})}{\kappa}} \tag{$\star$}.
%\\
%\nonumber & \|\widehat{\boldsymbol{\theta}}_{\ell_1}-\boldsymbol{\theta}\|_1 \leq \frac{6{s}\gamma_n}{\kappa}. 
\end{align}
\end{lemma}

In order to use Lemma \ref{RSC_thm}, we need to control $\gamma_n = \| \nabla \mathfrak{L}(\boldsymbol{\theta}) \|_\infty$. We have:
\begin{equation}
\nabla \mathfrak{L}(\boldsymbol{\theta}) = \frac{2}{n} \mathbf{X}^T(\mathbf{x}_1^n-\mathbf{X}\boldsymbol{\theta}),
\end{equation}
It is easy to check that by the uncorrelatedness of the innovations $w_k$'s, we have
\begin{equation}
\label{eq:expec=0}
\mathbb{E} \left[ \nabla \mathfrak{L}(\boldsymbol{\theta}) \right] = \frac{2}{n} \mathbb{E} \left[ \mathbf{X}^T (\mathbf{x}_1^n-\mathbf{X}\boldsymbol{\theta}) \right]=\frac{2}{n}\mathbb{E} \left[ \mathbf{X}^T \mathbf{w}_1^n \right]= \mathbf{0}.
\end{equation}
Eq. (\ref{eq:expec=0}) is known as the orthogonality principle. We next show that $ \nabla \mathfrak{L}(\boldsymbol{\theta})$ is concentrated around its mean. We can write 
\begin{equation*}
\left(\nabla \mathfrak{L}(\boldsymbol{\theta})\right)_i = \frac{2}{n} \mathbf{x}^{{n-i}^T}_{{-i+1}} \mathbf{w}_1^n,
\end{equation*}
and observe that the $j$th element in this expansion is of the form $y_j = x_{n-i-j+1}w_{n-j+1}$. It is easy to check that the sequence $y_1^n$ is a martingale with respect to the filtration given by
\begin{equation*}
\label{filtration}
\mathcal{F}_{j}=\sigma \left( \mathbf{x}_{-p+1}^{n-j+1} \right),
\end{equation*}
where $\sigma(\cdot)$ denote the sigma-field generated by the random variables $x_{-p+1}, x_{-p+2}, \cdots, x_{n-j+1}$. 
We use the following concentration result for sums of dependent random variables \cite{van_de_geer}:
\begin{lemma}
\label{hoeff_dep}
Fix $n\geq 1$. Let $Z_j$'s be sub-Gaussian $\mathcal{F}_j$-measurable random variables, satisfying for each $j=1,2,\cdots,n$,
\begin{equation*}
\mathbb{E}\left[Z_j|\mathcal{F}_{j-1}\right] = 0, \;\; \text{almost surely},
\end{equation*}
then there exists a constant $c$ such that for all $t>0$,
\begin{equation*}
\mathbb{P} \left( \left| \frac{1}{n}  \sum_{j=1}^n Z_j - \mathbb{E}[Z_j] \right| \geq t \right)\leq \exp\left(-\frac{nt^2}{c^2}\right).
\end{equation*}
\end{lemma}
\begin{proof}
This is a special case of Theorem 3.2 of \cite{van_de_geer} or Lemma 3.2 of \cite{geer2000empirical}, for sub-Gaussian-weighted sums of random variables. The constant $c$ depends on the sub-Gaussian constant of $Z_i$'s.
\end{proof}
Since $y_j$'s are a product of two independent sub-Gaussian random variables, they are sub-Gaussian as well. Lemma \ref{hoeff_dep} implies that
\vspace{-.2cm}  
\begin{equation}
\label{bound}
{\mathbb{P}\left( |\nabla \mathfrak{L}(\boldsymbol{\theta})_i|  \ge t \right) \leq  \exp\left(-\frac{nt^2}{c^2_0 \sigma^4_{\sf w}}\right).}
\end{equation}
where $c^2_0 := \frac{c^2}{\sigma_{\sf w}^4}$ is an absolute constant. By the union bound, we get:
\begin{equation}
\label{ubound}
{\mathbb{P}\Big(\left\| \nabla \mathfrak{L}(\boldsymbol{\theta})\right\|_\infty \ge t \Big) \leq  \exp\left(-\frac{t^2n}{c_0^2 \sigma^4_{\sf w}}+\log p\right).}
\end{equation}
Let $d_4$ be any positive integer. Choosing $t = c_0 \sigma^2_{\sf w}\sqrt{{1+d_4}}\sqrt{\frac{\log p}{n}}$, we get:
\vspace{-.2cm}
\begin{align}
\label{grad_bound}
\notag \mathbb{P}\left(\left\| \nabla \mathfrak{L}(\boldsymbol{\theta})\right\|_\infty \ge c_0 \sigma^2_{\sf w}\sqrt{{{1+d_4}}}\sqrt{\frac{\log p}{n}} \right) \leq \frac{2}{n^{d_4}}.
\end{align}
Hence, a choice of $\gamma_n = d_2 \sqrt{\frac{\log p}{n}}$ with $d_2 := c_0 \sigma^2_{\sf w}\sqrt{{1 + d_4}}$, satisfies $\gamma_n \ge \| \nabla \mathfrak{L}(\boldsymbol{\theta}) \|_\infty$ with probability at least $1 - \frac{2}{n^{d_4}}$. {Let $d_0:= \frac{(3+d_4)^2}{c_\eta^2}$} and {$d_1 = \frac{4 c'_\eta ( 3 + d_4)}{c_\eta}$}. Using Lemma \ref{lem:rsc}, the fact that {$n > s \max \{d_0 (\log p)^2, d_1 {(p \log p)^{1/2}}\}$} by hypothesis, and $p > n$ we have that the RSC of order $s$ hold for $\kappa = \frac{\sigma^2_{\sf w}}{16 \pi }$ with a probability at least $1- {\frac{2c_1}{p^{d_4}}} - \frac{1}{p^{d_4}}$. Combining these two assertions, the claim of Theorem 1 follows for $d_3 = 32 \pi c_0 \sqrt{1+d_4}$. \QEDB

\subsection{Proof of Theorem \ref{thm_OMP}}\label{prf:aromp}

The proof is mainly based on the following lemma, adopted from Theorem 2.1 of \cite{zhang_omp}, stating that the greedy procedure is successful in obtaining a reasonable $s^\star$-sparse approximation, if the cost function satisfies the RSC:
\begin{lemma}\label{prop_omp}
Let $s^\star$ be a constant such that
\begin{equation}
\label{sstar}
{s^\star \geq {4 \rho s}\log {20 \rho s}},
\end{equation}
and suppose that $\mathfrak{L}(\boldsymbol{\theta})$ satisfies RSC of order $s^\star$ with a constant $\kappa > 0$. Then, we have
\begin{equation*}
\left \|\widehat{\boldsymbol{\theta}}^{(s^\star)}_{{\sf OMP}}-\boldsymbol{\theta}_S \right \|_2 \leq \frac{\sqrt{6} \varepsilon_{s^\star}}{\kappa},
\end{equation*}
where $\eta_{s^\star}$ satisfies
\begin{equation}
\label{eps_bound}
\varepsilon_{s^\star} \leq \sqrt{s^\star+s} \|\nabla\mathfrak{L}(\boldsymbol{\theta}_S)\|_\infty .
\end{equation}
\end{lemma}
\begin{proof}
The proof is a specialization of the proof of Theorem 2.1 in \cite{zhang_omp} to our setting with the spectral spread ${\rho=1/4\eta^2}$.
\end{proof}

In order to use Lemma \ref{prop_omp}, we need to bound $\|\nabla\mathfrak{L}(\boldsymbol{\theta}_S)\|_\infty$. We have:
\begin{align*}
\mathbb{E} \left[\nabla\mathfrak{L}(\boldsymbol{\theta}_S)\right] &= \frac{1}{n}\mathbb{E} \left[\mathbf{X}^T(\mathbf{x}_1^n-\mathbf{X}\boldsymbol{\theta}_S)\right] = \frac{1}{n}\mathbb{E} \left[\mathbf{X}^T\mathbf{X} (\boldsymbol{\theta}-\boldsymbol{\theta}_S)\right]\\
& = \mathbf{R} (\boldsymbol{\theta}-\boldsymbol{\theta}_S) \le \frac{\sigma^2_{\sf w}}{2\pi\eta^2} \varsigma_s(\boldsymbol{\theta}) \mathbf{1},
\end{align*}
where in the second inequality we have used (\ref{eq:expec=0}), and the last inequality results from Corollary \ref{cor:eig_conv}. 
Let $d'_4$ be any positive integer. Using the result of Lemma \ref{hoeff_dep} together with the union bound yields:  
\begin{align*}
\resizebox{\columnwidth}{!}{$\displaystyle \mathbb{P}\left(\|\nabla\mathfrak{L}(\boldsymbol{\theta}_S)\|_\infty \geq  c_0 \sigma^2_{\sf w}\sqrt{{1+d'_4}} \sqrt{\frac{\log p}{n}} + \frac{\sigma^2_{\sf w}\varsigma_s(\boldsymbol{\theta})}{2 \pi \eta^2} \right) \leq \frac{2}{n^{d'_4}}$}.
\end{align*}
Hence, we get the following concentration result for $\varepsilon_{s^\star}$:
\begin{align}
\label{eps_prob}
\resizebox{\columnwidth}{!}{$\mathbb{P} \left(\varepsilon_{s^\star} \geq \sqrt{s^\star+s} \left(  c_0 \sigma^2_{\sf w}\sqrt{{1+d'_4}} \sqrt{\frac{\log p}{n}} + \frac{\sigma^2_{\sf w}\varsigma_s(\boldsymbol{\theta})}{2 \pi \eta^2}\right)\right)
\leq\frac{2}{n^{d'_4}}$}.
\end{align}
Noting that by (\ref{sstar}) we have $s^\star+s \le \frac{4 s \log s}{\eta^2} $. Let {$d'_0 = \frac{4(3+d'_4)^2}{\eta^2 c_\eta^2}$} and {$d'_{1} = \frac{16 c'_\eta ( 3 + d_4)}{c_\eta}$}. By the hypothesis of $\varsigma_s(\boldsymbol{\theta}) \le {A} s^{1- \frac{1}{\xi}}$ for some constant $A$, and invoking the results of Lemmas \ref{lem:rsc} and \ref{prop_omp}, we get:
\begin{align*}
\left \|\widehat{\boldsymbol{\theta}}^{(s^\star)}_{{\sf OMP}}-\boldsymbol{\theta}_S \right\|_2 &\leq d'_{2} \sqrt{\frac{s \log s \log p}{n}} + d''_{2} \sqrt{s\log s} \varsigma_s(\boldsymbol{\theta})\\
&\leq  d'_{2} \sqrt{\frac{s \log s \log p}{n}} + d''_{2} \frac{\sqrt{\log s}}{s^{\frac{1}{\xi}-\frac{3}{2}}},
\end{align*}
where $d'_{2} = \frac{16 \pi c_0 \sqrt{24 (1+d'_4)}}{\eta}$ and $d''_{2} = \frac{A }{ \pi \eta^3}$, 
with probability at least  $1- {\frac{2c_1}{p^{d'_4}}} - {\frac{1}{p^{d'_4}}} -\frac{2}{n^{d'_4}}$. Finally, we have: 
\begin{align*}
\left \|\widehat{\boldsymbol{\theta}}^{(s^\star)}_{{\sf OMP}}-\boldsymbol{\theta}\right\|_2 &= \left \|\widehat{\boldsymbol{\theta}}^{(s^\star)}_{{\sf OMP}}-\boldsymbol{\theta}_S +\boldsymbol{\theta}_S -\boldsymbol{\theta} \right \|_2 \\
& \leq \left \|\widehat{\boldsymbol{\theta}}^{(s^\star)}_{{\sf OMP}}-\boldsymbol{\theta}_S \right \|_2 + \|\boldsymbol{\theta}_S-\boldsymbol{\theta}\|_2.
\end{align*}
Choosing $d'_{3} = 2 d''_{2}$ completes the proof. \QEDB

\subsection{Proof of Proposition \ref{thm:ar_minimax}}
\label{prf:minimax}
\noindent Consider the event defined by
\begin{align*}
\mathcal{A}:=\left \{\max_{i,j}|\widehat{R}_{ij}-R_{ij}| \leq \tau \right \}.
\end{align*}
Eq. (\ref{eq:max}) in the proof of Lemma \ref{lem:re2} implies that:
\begin{align*}
\mathbb{P}(\mathcal{A}^c) \leq c_1 p^2 (n+p) \exp \left( - \frac{ c_2 \sqrt{\tau n}}{\frac{c_3 \sigma_{\sf w}^4 (n+p)}{\tau^{3/2} n^{3/2}} + \sigma_{\sf w}}\right).
\end{align*}
By choosing $\tau$ as in the proof of Theorem \ref{thm:ar_1}, we have
\begin{align*}
\mathcal{R}^2_{\sf est} (\widehat{\boldsymbol{\theta}}_{\sf minimax}) & \leq \mathcal{R}^2_{\sf est} (\widehat{\boldsymbol{\theta}}_{\ell_1}) = \sup_{\mathcal{H}} \left(\mathbb{E}\left[\|\widehat{\boldsymbol{\theta}}_{\ell_1}-\boldsymbol{\theta}\|_2^2\right]\right)\\
& \leq \mathbb{P}(\mathcal{A})  d^2_3 {\frac{s \log p}{n}} +\sup_{\mathcal{H}} \mathbb{E}_{{A}^c} \left [  \|\widehat{\boldsymbol{\theta}}_{\ell_1}-{\boldsymbol{\theta}}\|_2^2\right]\\
&\leq \resizebox{0.7\columnwidth}{!}{$d^2_3 {\frac{s \log p}{n}}+ \displaystyle 8(1-\eta)^2 c_1 \exp \left(-\frac{ c_2 \sqrt{\tau n}}{\frac{c_3 \sigma_{\sf w}^4 (n+p)}{\tau^{3/2} n^{3/2}} + \sigma_{\sf w}} + 3 \log p\right)$},
\end{align*}
where the second inequality follows from Theorem \ref{thm:ar_1}, and the third inequality follows from the fact that $\| \widehat{\boldsymbol{\theta}}_{\ell_1} - \boldsymbol{\theta} \|_2^2 \le 4 (1-\eta)^2$ by the sufficient stability assumption. For {$n > s \max \{ d_0 (\log p)^2, d_1 {(p \log p)^{1/2}\}}$}, the first term will be the dominant, and thus we get $\mathcal{R}_{\sf est} (\widehat{\boldsymbol{\theta}}_{\sf minimax}) \le 2 d_3 \sqrt{\frac{s \log p}{n}}$, for large enough $n$.

As for a lower bound on $\mathcal{R}_{\sf est} (\widehat{\boldsymbol{\theta}}_{\sf minimax})$, we take the approach of \cite{goldenshluger2001nonasymptotic} by constructing a family of AR processes with sparse parameters $\boldsymbol{\theta}$ for which the minimax risk is optimal modulo constants. In our construction, we assume that the innovations are Gaussian. The key element of the proof is the Fano's inequality:
\begin{lemma}[Fano's Inequality]
\label{lem:ar_fano_ineq}
Let $\mathcal{Z}$ be a class of densities with a subclass $\mathcal{Z}^\star$ of densities $f_{\boldsymbol{\theta}_i}$, parameterized by $\boldsymbol{\theta}_i$, for $i \in \{0,\cdots,2^M \}$. Suppose that for any two distinct $\boldsymbol{\theta}_1, \boldsymbol{\theta}_2 \in \mathcal{Z}^\star$, 
$\mathcal{D}_{\sf KL}(f_{\boldsymbol{\theta}_1}\| f_{\boldsymbol{\theta}_2}) \leq \beta$ for some constant $\beta$. Let $\widehat{\boldsymbol{\theta}}$ be an estimate of the parameters. Then
\vspace{-.2cm}
\begin{equation}
\label{eq:fano}
\sup_j \mathbb{P}(\widehat{\boldsymbol{\theta}}\neq {\boldsymbol{\theta}_j}|H_j) \geq 1 - \frac{\beta+\log2}{M},
\end{equation}
where $H_j$ denotes the hypothesis that $\boldsymbol{\theta}_j$ is the true parameter, and induces the probability measure $\mathbb{P}(.|H_j)$.
\end{lemma}

{Consider a class $\mathcal{Z}$ of AR processes with $s$-sparse parameters over any subset $S \subset \{1,2,\cdots,p\}$ satisfying $|S| =s$,} with parameters given by
\begin{equation}
\label{eq:ar_minimax_params}
\theta_\ell = \pm e^{-m} \mathbbm{1}_S(\ell),
\end{equation}
where $m$ remains to be chosen. We also add the all zero vector $\boldsymbol{\theta}$ to $\mathcal{Z}$. For a fixed $S$,  we have $2^s + 1$ such parameters forming a subfamily $\mathcal{Z}_S$. Consider the maximal collection of ${p \choose s}$ subsets $S$ for which any two subsets differ in at least $s/4$ indices. The size of this collection can be identified by $A(p, \frac{s}{4}, s)$ in coding theory, where $A(n,d,w)$ represents the maximum size of a binary code of length $n$ with minimum distance $d$ and constant weight $w$ \cite{macwilliams1977theory}. We have
\[
A(p, {\textstyle \frac{s}{4} }, s) \ge \frac{p^{\frac{7}{8}s - 1}}{s!},
\]
for large enough $p$ (See Theorem 6 in \cite{graham1980lower}). Also, by the Gilbert-Varshamov bound \cite{macwilliams1977theory}, there exists a subfamily $\mathcal{Z}_S^\star \subset \mathcal{Z}_S$, of cardinality $|\mathcal{Z}_S^\star| \geq 2^{\lfloor s/8 \rfloor}+1$, such that any two distinct $\boldsymbol{\theta}_1,\boldsymbol{\theta}_2 \in \mathcal{Z}_S^\star$ differ at least in $s/16$ components. Thus  for $\boldsymbol{\theta}_1, \boldsymbol{\theta}_2 \in \mathcal{Z}^\star :=  \resizebox{!}{0.3cm}{$\displaystyle \bigcup_S$} \mathcal{Z}^\star_S$, we have
\vspace{-.2cm}
\begin{equation}
\label{eq:ar_theta_lowerbd}
\|\boldsymbol{\theta}_1-\boldsymbol{\theta}_2\|_2 \geq \frac{1}{4} \sqrt{s} e^{-m} =: \alpha,
\end{equation}
and $| \mathcal{Z}^\star | \ge \frac{p^{\frac{7}{8}s - 1}}{s!} 2^{\lfloor s/8 \rfloor}$. For an arbitrary estimate $\widehat{\boldsymbol{\theta}}$, consider the testing problem between the $\frac{p^{\frac{7}{8}s - 1}}{s!} 2^{\lfloor s/8 \rfloor}$ hypotheses $H_j: \boldsymbol{\theta}=\boldsymbol{\theta}_j \in \mathcal{Z}^\star$, using the minimum distance decoding strategy. Using Markov's inequality we have
\begin{align}
\label{eq:ar_sup_lowerbd}
\notag \sup_{\mathcal{Z}} \mathbb{E}\left[\|\widehat{\boldsymbol{\theta}}-{\boldsymbol{\theta}}\|_2\right] &\geq \sup_{\mathcal{Z}^\star} \mathbb{E}\left [\|\widehat{\boldsymbol{\theta}}-{\boldsymbol{\theta}}\|_2\right]\\
\notag &\geq \frac{\alpha}{2} \sup_{\mathcal{Z}^\star}\mathbb{P}\left(\|\widehat{\boldsymbol{\theta}}-{\boldsymbol{\theta}}\|_2 \geq \frac{\alpha}{2} \right)\\
& = \frac{\alpha}{2}\sup_{j} \mathbb{P}\left(\widehat{\boldsymbol{\theta}}\neq {\boldsymbol{\theta}_j}|H_j \right).
\end{align}
Let $f_{\boldsymbol{\theta}_j}$ denote joint probability distribution of $\{x_k\}_{k=1}^n$ conditioned on $\{x_k\}_{k=-p+1}^0$ under the hypothesis $H_j$. Using the Gaussian assumption on the innovations, for $i \neq j$, we have
\begin{align}
\label{eq:ar_kl_bound}
\notag
& \mathcal{D}_{\sf KL}(f_{\boldsymbol{\theta}_i}\|f_{\boldsymbol{\theta}_j}) \leq \sup_{i\neq j} \mathbb{E}\left[\log \frac{f_{\boldsymbol{\theta}_i}}{f_{\boldsymbol{\theta}_j}} |H_i\right]\\
\notag & \leq \sup_{i \neq j}\resizebox{.89\columnwidth}{!}{$\displaystyle \mathbb{E}\left[-\frac{1}{2\sigma^2_{\sf w}} \sum_{k=1}^n \left( \left(x_k-\boldsymbol{\theta}_i'\mathbf{x}_{k-p}^{k-1} \right)^2 - \left(x_k-\boldsymbol{\theta}_j'\mathbf{x}_{k-p}^{k-1} \right)^2 \right)  \Big| H_i\right]$}\\
\notag & \leq \sup_{i \neq j} \frac{n}{2\sigma^2_{\sf w}} \mathbb{E}\left[ \left((\boldsymbol{\theta}_i-\boldsymbol{\theta}_j)'\mathbf{x}_{k-p}^{k-1}\right)^2 \Big| H_i \right]\\
\notag & = \frac{n}{2\sigma^2_{\sf w}}\sup_{i \neq j} (\boldsymbol{\theta}_i-\boldsymbol{\theta}_j)' \mathbf{R} (\boldsymbol{\theta}_i-\boldsymbol{\theta}_j)\\
& \leq \frac{n \lambda_{\sf max}}{2\sigma^2_{\sf w}}\sup_{i \neq j} \|\boldsymbol{\theta}_i-\boldsymbol{\theta}_j\|_2^2 \leq \frac{n s e^{-2m}}{64 \pi \eta^2} =: \beta.
\end{align}

Using Lemma \ref{lem:ar_fano_ineq}, (\ref{eq:ar_theta_lowerbd}), (\ref{eq:ar_sup_lowerbd}) and (\ref{eq:ar_kl_bound}) yield:
\begin{equation*}
\sup_{\mathcal{Z}} \mathbb{E}\left[\|\widehat{\boldsymbol{\theta}}-{\boldsymbol{\theta}}\|_2 \right] \geq \frac{\sqrt{s}  e^{-m}}{8}\left(1-\frac{2\left( \frac{  n s e^{-2m}}{64 \pi \eta^2}+\log 2\right)}{s \log p}\right).
\end{equation*}
for $p$ large enough so that $\log p \ge \frac{\log s - \frac{9}{8}}{\frac{3}{8} - \frac{1}{s}}$. Choosing $m = \frac{1}{2} \log \left ( \frac{n}{8 \pi \eta^2 \log p} \right )$ gives us the claim of Proposition \ref{thm:ar_minimax} with $L = \frac{d_3}{\eta \sqrt{2\pi}}$ for large enough $s$ and $p$ such that $s \log p \ge \log (256)$. The hypothesis of $s  \le \frac{1- \eta}{\sqrt{8\pi} \eta} \sqrt{\frac{n}{\log p}}$ guarantees that for all $\boldsymbol{\theta}\in \mathcal{Z}^\star$, we have $\|\boldsymbol{\theta}\|_1 \leq 1-\eta$. \QEDB

%In order to calculate the one-step prediction error assume that $\widehat{\boldsymbol{\theta}}_{\ell_1}$ is padded with zeros (possibly with infinitely many) so that it matches the order of $\boldsymbol{\theta}$. Now we have
%\begin{equation*}
%e_{n+1} =  \underbrace{\left( {\boldsymbol{\theta}}^{(p)}-\widehat{\boldsymbol{\theta}}_{{\ell_1}}^{(p)} \right)'x_{n-p+1}^{n}}_{\mathcal{E}_1} + \underbrace{\sum_{i=p+1}^{\infty}{\boldsymbol{\theta}}_ix_{n-i+1}}_{\mathcal{E}_2}+w_{n+1}.
%\end{equation*}
%Therefore
%\begin{equation}
%\mathbb{E}\left[e_{n+1}^2 \right] = \mathbb{E}\left[\left(\mathcal{E}_1+\mathcal{E}_2\right)^2\right]+1 \leq 2 \left(\mathbb{E}\left[\mathcal{E}_1^2\right] +\mathbb{E}\left[\mathcal{E}_2^2\right] \right)+1.
%\end{equation}
%We will now bound each term on the right hand side separately. We have
%
%
%\begin{align}
%\label{eq:ar_pred_mse}
%\notag \mathbb{E}\left[\mathcal{E}_2^2\right] & = \mathbb{E}\left[ \left( \left( {\boldsymbol{\theta}}^{(p)}-\widehat{\boldsymbol{\theta}}_{{\ell_1}} \right)'x_{n-p+1}^{n}\right)^2 \right]\\
%& \leq  \| {\boldsymbol{\theta}}^{(p)}-\widehat{\boldsymbol{\theta}}_{{\ell_1}}\|_2^2
%\end{align}

%\section{Generalization to Stable ]AR Models}
%\label{app:ar_generalization}

\vspace{-3mm}
\subsection{Generalization to stable AR processes}
We consider relaxing the sufficient stability assumption of $\| \boldsymbol{\theta} \|_1 \le 1- \eta < 1$ to $\boldsymbol{\theta}$ being in the set of stable AR processes. Given that the set of all stable AR processes is not necessarily convex, the LASSO and OMP estimates cannot be obtained by convex optimization techniques. Nevertheless, the results of Theorems \ref{thm:ar_1} and \ref{thm_OMP} can be generalized to the case of stable AR models:

\begin{corollary}\label{cor:ar_1} The claims of Theorems \ref{thm:ar_1} and \ref{thm_OMP} hold when $\boldsymbol{\Theta}$ is replaced by the set of stable AR processes, except for possibly slightly different constants.
\end{corollary}

\begin{proof} Note that the stability of the process guarantees boundedness of the power spectral density. The result follows by simply replacing the bounds $\left [\frac{\sigma^2_{\sf w}}{8 \pi}, \frac{\sigma^2_{\sf w}}{2 \pi \eta^2} \right]$ on the singular values of the covariance matrix $\mathbf{R}$ in Corollary \ref{cor:eig_conv} by $[\inf_\omega S(\omega), \sup_\omega S(\omega)]$.
\end{proof}

\vspace{-3mm}
\section{Statistical Tests for Goodness-of-Fit}
\label{app:ar_tests}

In this appendix, we will give an overview of the statistical goodness-of-fit tests for assessing the accuracy of the AR model estimates. A detailed treatment can be found in  \cite{lehmann1986testing}.
%{These tests assume knowledge of the reference distribution $F_0$ for the null hypothesis. For simulation purposes we have used a two-fold cross-validation estimating $F_0$ using half of the data as training.}
\subsection{Residue-based tests}
Let $\widehat{\boldsymbol{\theta}}$ be an estimate of the parameters of the process. The residues (estimated innovations) of the process based on $\widehat{\boldsymbol{\theta}}$ are given by
\vspace{-.2cm}
\begin{equation*}
e_k = x_k - \widehat{\boldsymbol{\theta}}\mathbf{x}_{k-p}^{k-1}, \quad\quad i=1,2,\cdots,n. 
\end{equation*}
The main idea behind most of the available statistical tests is to quantify how close the sequence $\{e_i\}_{i=1}^{n}$ is to an i.i.d. realization of a known distribution $F_0$ which is most likely absolutely continuous . Let us denote the empirical distribution of the $n$-samples by $\widehat{F}_n$. If the samples are generated from $F_0$ the Glivenko-Cantelli theorem suggests that:
\begin{equation*}
\sup_t \;|\widehat{F}_n(t)-F_0(t)| \stackrel{\sf a.s.}\longrightarrow 0.
\vspace{-.2cm}
\end{equation*}
\noindent That is, for large $n$ the empirical distribution $\widehat{F}_n$ is uniformly close to $F_0$. The Kolmogorov-Smirnov (KS) test, Cram\'{e}r-von Mises (CvM) criterion and the Anderson-Darling (AD) test are three measures of discrepancy between $\widehat{F}_n$ and $F_0$ which are easy to compute and are sufficiently discriminant against alternative distributions. More specifically, the limiting distribution of the following three random variables are known:
\noindent The KS test statistic
\begin{equation*}
K_n:=\sup_t \; |\widehat{F}_n(t)-F_0(t)|,
\end{equation*}
the CvM statistic
\begin{equation*}
C_n:= \int \big(\widehat{F}_n(t)-F_0(t)\big)^2 dF_0(t),
\end{equation*}
and the AD statistic
\begin{equation*}
A_n:= \int \frac{\big(\widehat{F}_n(t)-F_0(t)\big)^2}{F_0(t)\left(1-F_0(t)\right)}dF_0(t).
\end{equation*}
For large values of $n$, the Glivenko-Cantelli theorem also suggests that these statistics should be small. A simple calculation leads to the following equivalent for the statistics:
\[
K_n = \max_{1\leq i \leq n} \max \left\{\left|\frac{i}{n}-F_0(e_i) \right|, \left|\frac{i-1}{n}-F_0(e_i) \right| \right\},
\]
\[
nC_n = \frac{1}{12n}+\sum_{i=1}^n \left(F_0(e_i)-\frac{2i-1}{2n} \right)^2,
\]
and
\vspace{-.2cm}
\[
nA_n = -n-\frac{1}{n}\sum_{i=1}^n\left(2i-1\right)\Big( \log F_0(e_i)+ \log\left(1-F_0(e_i)\Big)\right).
\]
\vspace{-.7cm}

\subsection{\textcolor{black}{Spectral domain} tests for Gaussian AR processes}
The aforementioned KS, CvM and AD tests all depend on the distribution of the innovations. For Gaussian AR processes, the spectral  versions of these tests are introduced in \cite{anderson1997goodness}. These tests are based on the similarities of the periodogram of the data and the estimated power-spectral density of the process. The key idea is summarized in the following lemma:
\begin{lemma}
\label{lemma:ar_spec_CvM}
Let $S(\omega)$ be the (normalized) power-spectral density of stationary process with bounded spectral spread, and $\widehat{S}_n(\omega)$ be the periodogram of the $n$ samples of a realization of such a process, then for all $\omega$ we have:
\begin{equation}
\label{eq:ar_spec_gaussian}
\sqrt{n}\left(2\int_{0}^{\omega}\left(\widehat{S}_n(\lambda)- S(\lambda)\right) d\lambda \right) \stackrel{\sf d.}\longrightarrow \mathcal{Z}(\omega),
\end{equation}
where $\mathcal{Z}(\omega)$ is a zero-mean Gaussian process.
%\[
%\mathbb{E}[\mathcal{Z}(\omega)\mathcal{Z}(\lambda) ] = 4\pi 
%\]
\end{lemma}
The explicit formula for the covariance function of $\mathcal{Z}(.)$ is calculated in \cite{anderson1997goodness}. Lemma \ref{lemma:ar_spec_CvM} suggests that for a good estimate $\widehat{\boldsymbol{\theta}}$ which admits a power spectral density $S(\omega;\widehat{\boldsymbol{\theta}})$,  one should get a (\textit{close} to) Gaussian process replacing $S(\omega)$ with $S(\omega;\widehat{\boldsymbol{\theta}})$ in (\ref{eq:ar_spec_gaussian}). The spectral form of the CvM, KS and AD statistics can thus be characterized given an estimate $\widehat{\boldsymbol{\theta}}$.

\vspace{-.2cm}
\section*{Acknowledgment}
This material is based upon work supported in part by the National Science Foundation under Grant No. 1552946.

{
\small
\bibliographystyle{IEEEtran}
\bibliography{Sparse_AR}

\begin{thebibliography}{10}
\providecommand{\url}[1]{#1}
\csname url@samestyle\endcsname
\providecommand{\newblock}{\relax}
\providecommand{\bibinfo}[2]{#2}
\providecommand{\BIBentrySTDinterwordspacing}{\spaceskip=0pt\relax}
\providecommand{\BIBentryALTinterwordstretchfactor}{4}
\providecommand{\BIBentryALTinterwordspacing}{\spaceskip=\fontdimen2\font plus
\BIBentryALTinterwordstretchfactor\fontdimen3\font minus
  \fontdimen4\font\relax}
\providecommand{\BIBforeignlanguage}[2]{{%
\expandafter\ifx\csname l@#1\endcsname\relax
\typeout{** WARNING: IEEEtran.bst: No hyphenation pattern has been}%
\typeout{** loaded for the language `#1'. Using the pattern for}%
\typeout{** the default language instead.}%
\else
\language=\csname l@#1\endcsname
\fi
#2}}
\providecommand{\BIBdecl}{\relax}
\BIBdecl

\bibitem{kazemipour-ciss}
A.~Kazemipour, B.~Babadi, and M.~Wu, ``Sufficient conditions for stable
  recovery of sparse autoregressive models,'' in \emph{50th Annual Conference
  on Information Sciences and Systems (CISS), March 16--18, Princeton, NJ},
  2016.

\bibitem{sang2015simultaneous}
H.~Sang and Y.~Sun, ``Simultaneous sparse model selection and coefficient
  estimation for heavy-tailed autoregressive processes,'' \emph{Statistics},
  vol.~49, no.~1, pp. 187--208, 2015.

\bibitem{farokhi2014vehicular}
K.~Farokhi~Sadabadi, ``Vehicular traffic modelling, data assimilation,
  estimation and short term travel time prediction,'' Ph.D. dissertation,
  University of Maryland, College Park, 2014.

\bibitem{ahmed1982application}
S.~A. Ahmed and A.~R. Cook, \emph{Application of time-series analysis
  techniques to freeway incident detection}, 1982, no. 841.

\bibitem{ahmed1979analysis}
M.~S. Ahmed and A.~R. Cook, \emph{Analysis of freeway traffic time-series data
  by using Box-Jenkins techniques}, 1979, no. 722.

\bibitem{barcelo2010travel}
J.~Barcel{\'o}, L.~Montero, L.~Marqu{\'e}s, and C.~Carmona, ``Travel time
  forecasting and dynamic origin-destination estimation for freeways based on
  bluetooth traffic monitoring,'' \emph{Transportation Research Record: Journal
  of the Transportation Research Board}, no. 2175, pp. 19--27, 2010.

\bibitem{clark2003traffic}
S.~Clark, ``Traffic prediction using multivariate nonparametric regression,''
  \emph{Journal of transportation engineering}, vol. 129, no.~2, pp. 161--168,
  2003.

\bibitem{robinson2003time}
P.~M. Robinson, \emph{Time series with long memory}.\hskip 1em plus 0.5em minus
  0.4em\relax Oxford University Press, 2003.

\bibitem{akaike1969fitting}
H.~Akaike, ``Fitting autoregressive models for prediction,'' \emph{Annals of
  the institute of Statistical Mathematics}, vol.~21, no.~1, pp. 243--247,
  1969.

\bibitem{poskitt2007autoregressive}
D.~S. Poskitt, ``Autoregressive approximation in nonstandard situations: the
  fractionally integrated and non-invertible cases,'' \emph{Annals of the
  Institute of Statistical Mathematics}, vol.~59, no.~4, pp. 697--725, 2007.

\bibitem{shibata1980asymptotically}
R.~Shibata, ``Asymptotically efficient selection of the order of the model for
  estimating parameters of a linear process,'' \emph{The Annals of Statistics},
  pp. 147--164, 1980.

\bibitem{galbraith1997some}
J.~W. Galbraith and V.~Zinde-Walsh, ``On some simple, autoregression-based
  estimation and identification techniques for arma models,''
  \emph{Biometrika}, vol.~84, no.~3, pp. 685--696, 1997.

\bibitem{galbraith2001autoregression}
J.~Galbraith and V.~Zinde-Walsh, ``Autoregression-based estimators for arfima
  models,'' CIRANO, Tech. Rep., 2001.

\bibitem{ing2005order}
C.-K. Ing and C.-Z. Wei, ``Order selection for same-realization predictions in
  autoregressive processes,'' \emph{The Annals of Statistics}, vol.~33, no.~5,
  pp. 2423--2474, 2005.

\bibitem{baddour2005autoregressive}
K.~E. Baddour and N.~C. Beaulieu, ``Autoregressive modeling for fading channel
  simulation,'' \emph{IEEE Transactions on Wireless Communications}, vol.~4,
  no.~4, pp. 1650--1662, 2005.

\bibitem{mann1999oscillatory}
M.~E. Mann and J.~Park, ``Oscillatory spatiotemporal signal detection in
  climate studies: A multiple-taper spectral domain approach,'' \emph{Advances
  in geophysics}, vol.~41, pp. 1--132, 1999.

\bibitem{akaike1973maximum}
H.~Akaike, ``Maximum likelihood identification of gaussian autoregressive
  moving average models,'' \emph{Biometrika}, vol.~60, no.~2, pp. 255--265,
  1973.

\bibitem{akaike1970statistical}
------, ``Statistical predictor identification,'' \emph{Annals of the Institute
  of Statistical Mathematics}, vol.~22, no.~1, pp. 203--217, 1970.

\bibitem{schwarz1978estimating}
G.~Schwarz, ``Estimating the dimension of a model,'' \emph{The annals of
  statistics}, vol.~6, no.~2, pp. 461--464, 1978.

\bibitem{wang2007regression}
H.~Wang, G.~Li, and C.-L. Tsai, ``Regression coefficient and autoregressive
  order shrinkage and selection via the lasso,'' \emph{Journal of the Royal
  Statistical Society: Series B (Statistical Methodology)}, vol.~69, no.~1, pp.
  63--78, 2007.

\bibitem{goldenshluger2001nonasymptotic}
A.~Goldenshluger and A.~Zeevi, ``Nonasymptotic bounds for autoregressive time
  series modeling,'' \emph{Annals of statistics}, pp. 417--444, 2001.

\bibitem{nardi2011autoregressive}
Y.~Nardi and A.~Rinaldo, ``Autoregressive process modeling via the lasso
  procedure,'' \emph{Journal of Multivariate Analysis}, vol. 102, no.~3, pp.
  528--549, 2011.

\bibitem{donoho2006compressed}
D.~L. Donoho, ``Compressed sensing,'' \emph{IEEE Transactions on Information
  Theory}, vol.~52, no.~4, pp. 1289--1306, 2006.

\bibitem{candes2006compressive}
E.~J. Cand{\`e}s, ``Compressive sampling,'' in \emph{Proceedings of the
  International Congress of Mathematicians Madrid, August 22--30}, 2006, pp.
  1433--1452.

\bibitem{candes2008introduction}
E.~J. Cand{\`e}s and M.~B. Wakin, ``An introduction to compressive sampling,''
  \emph{IEEE Signal Processing Magazine}, vol.~25, no.~2, pp. 21--30, 2008.

\bibitem{rudelson2008sparse}
M.~Rudelson and R.~Vershynin, ``On sparse reconstruction from fourier and
  gaussian measurements,'' \emph{Communications on Pure and Applied
  Mathematics}, vol.~61, no.~8, pp. 1025--1045, 2008.

\bibitem{baraniuk2008simple}
R.~Baraniuk, M.~Davenport, R.~DeVore, and M.~Wakin, ``A simple proof of the
  restricted isometry property for random matrices,'' \emph{Constructive
  Approximation}, vol.~28, no.~3, pp. 253--263, 2008.

\bibitem{zhao2006model}
P.~Zhao and B.~Yu, ``On model selection consistency of lasso,'' \emph{Journal
  of Machine Learning Research}, vol.~7, no. Nov, pp. 2541--2563, 2006.

\bibitem{raskutti2010restricted}
G.~Raskutti, M.~J. Wainwright, and B.~Yu, ``Restricted eigenvalue properties
  for correlated gaussian designs,'' \emph{The Journal of Machine Learning
  Research}, vol.~11, pp. 2241--2259, 2010.

\bibitem{Toeplitz}
J.~Haupt, W.~U. Bajwa, G.~Raz, and R.~Nowak, ``Toeplitz compressed sensing
  matrices with applications to sparse channel estimation,'' \emph{IEEE Trans.
  on Information Theory}, vol.~56, no.~11, pp. 5862--5875, 2010.

\bibitem{rauhut2012restricted}
H.~Rauhut, J.~Romberg, and J.~A. Tropp, ``Restricted isometries for partial
  random circulant matrices,'' \emph{Applied and Computational Harmonic
  Analysis}, vol.~32, no.~2, pp. 242--254, 2012.

\bibitem{loh2012high}
P.-L. Loh and M.~J. Wainwright, ``High-dimensional regression with noisy and
  missing data: Provable guarantees with non-convexity,'' \emph{The Annals of
  Statistics}, vol.~40, no.~3, pp. 1637--1664, 2012.

\bibitem{han2013transition}
F.~Han and H.~Liu, ``Transition matrix estimation in high dimensional time
  series.'' in \emph{ICML (2)}, 2013, pp. 172--180.

\bibitem{negahban2011estimation}
S.~Negahban and M.~J. Wainwright, ``Estimation of (near) low-rank matrices with
  noise and high-dimensional scaling,'' \emph{The Annals of Statistics}, pp.
  1069--1097, 2011.

\bibitem{wong2016regularized}
K.~C. Wong, A.~Tewari, and Z.~Li, ``Regularized estimation in high dimensional
  time series under mixing conditions,'' \emph{arXiv preprint
  arXiv:1602.04265}, 2016.

\bibitem{basu2015regularized}
S.~Basu and G.~Michailidis, ``Regularized estimation in sparse high-dimensional
  time series models,'' \emph{The Annals of Statistics}, vol.~43, no.~4, pp.
  1535--1567, 2015.

\bibitem{wu2016performance}
W.-B. Wu and Y.~N. Wu, ``Performance bounds for parameter estimates of
  high-dimensional linear models with correlated errors,'' \emph{Electronic
  Journal of Statistics}, vol.~10, no.~1, pp. 352--379, 2016.

\bibitem{bickel2009simultaneous}
P.~J. Bickel, Y.~Ritov, and A.~B. Tsybakov, ``Simultaneous analysis of lasso
  and dantzig selector,'' \emph{The Annals of Statistics}, pp. 1705--1732,
  2009.

\bibitem{stoica1997introduction}
P.~Stoica and R.~L. Moses, \emph{Introduction to spectral analysis}.\hskip 1em
  plus 0.5em minus 0.4em\relax Prentice hall Upper Saddle River, 1997, vol.~1.

\bibitem{haykin2008adaptive}
S.~S. Haykin, \emph{Adaptive filter theory}.\hskip 1em plus 0.5em minus
  0.4em\relax Pearson Education India, 2008.

\bibitem{burg1967maximum}
J.~P. Burg, ``Maximum entropy spectral analysis.'' in \emph{37th Annual
  International Meeting.}\hskip 1em plus 0.5em minus 0.4em\relax Society of
  Exploration Geophysics, 1967.

\bibitem{marple1987digital}
S.~L. Marple~Jr, ``Digital spectral analysis with applications,''
  \emph{Englewood Cliffs, NJ, Prentice-Hall, Inc., 1987, 512 p.}, vol.~1, 1987.

\bibitem{needell2009cosamp}
D.~Needell and J.~A. Tropp, ``{CoSaMP}: Iterative signal recovery from
  incomplete and inaccurate samples,'' \emph{Applied and Computational Harmonic
  Analysis}, vol.~26, no.~3, pp. 301--321, 2009.

\bibitem{percival1993spectral}
D.~B. Percival and A.~T. Walden, \emph{Spectral analysis for physical
  applications}.\hskip 1em plus 0.5em minus 0.4em\relax Cambridge University
  Press, 1993.

\bibitem{tibshirani1996regression}
R.~Tibshirani, ``Regression shrinkage and selection via the lasso,''
  \emph{Journal of the Royal Statistical Society. Series B (Methodological)},
  pp. 267--288, 1996.

\bibitem{wainwright2009sharp}
M.~J. Wainwright, ``Sharp thresholds for high-dimensional and noisy sparsity
  recovery using-constrained quadratic programming (lasso),'' \emph{IEEE
  transactions on information theory}, vol.~55, no.~5, pp. 2183--2202, 2009.

\bibitem{knight2000asymptotics}
K.~Knight and W.~Fu, ``Asymptotics for lasso-type estimators,'' \emph{Annals of
  statistics}, pp. 1356--1378, 2000.

\bibitem{meinshausen2006high}
N.~Meinshausen and P.~B{\"u}hlmann, ``High-dimensional graphs and variable
  selection with the lasso,'' \emph{The annals of statistics}, pp. 1436--1462,
  2006.

\bibitem{OMP}
Y.~C. Pati, R.~Rezaiifar, and P.~Krishnaprasad, ``Orthogonal matching pursuit:
  Recursive function approximation with applications to wavelet
  decomposition,'' in \emph{Conference Record of The Twenty-Seventh Asilomar
  Conference on Signals, Systems and Computers}.\hskip 1em plus 0.5em minus
  0.4em\relax IEEE, 1993, pp. 40--44.

\bibitem{zhang_omp}
T.~Zhang, ``Sparse recovery with orthogonal matching pursuit under {RIP},''
  \emph{IEEE Transactions on Information Theory}, vol.~57, no.~9, pp.
  6215--6221, 2011.

\bibitem{bruckstein2009sparse}
A.~M. Bruckstein, D.~L. Donoho, and M.~Elad, ``From sparse solutions of systems
  of equations to sparse modeling of signals and images,'' \emph{SIAM review},
  vol.~51, no.~1, pp. 34--81, 2009.

\bibitem{candes2006modern}
E.~J. Candes, ``Modern statistical estimation via oracle inequalities,''
  \emph{Acta numerica}, vol.~15, pp. 257--325, 2006.

\bibitem{d1986goodness}
R.~B. D'Agostino, \emph{Goodness-of-fit-techniques}.

\bibitem{johansen1995likelihood}
S.~Johansen, ``Likelihood-based inference in cointegrated vector autoregressive
  models,'' \emph{OUP Catalogue}, 1995.

\bibitem{anderson1997goodness}
T.~W. Anderson, ``Goodness-of-fit tests for autoregressive processes,''
  \emph{Journal of time series analysis}, vol.~18, no.~4, pp. 321--339, 1997.

\bibitem{cushing}
\BIBentryALTinterwordspacing
``Cushing, ok wti spot price fob dataset,'' (Date last accessed
  14-December-2015). [Online]. Available:
  \url{http://www.eia.gov/dnav/pet/hist/LeafHandler.ashx?n=PET&s=RWTC&f=D}
\BIBentrySTDinterwordspacing

\bibitem{ritis1}
\BIBentryALTinterwordspacing
``Regional integrated transportation information system (ritis),'' (Date last
  accessed 27-December-2015). [Online]. Available: \url{https://ritis.org}
\BIBentrySTDinterwordspacing

\bibitem{ritis2}
\BIBentryALTinterwordspacing
``Regional integrated transportation information system (ritis),'' (Date last
  accessed 27-December-2015). [Online]. Available:
  \url{http://i95coalition.org/projects/regional-integrated-transportation-information-system-ritis}
\BIBentrySTDinterwordspacing

\bibitem{kazemipour}
A.~Kazemipour, B.~Babadi, and M.~Wu, ``Sparse estimation of self-exciting point
  processes with application to {LGN} neural modeling,'' in \emph{2014 IEEE
  Global Conference on Signal and Information Processing (GlobalSIP)}.\hskip
  1em plus 0.5em minus 0.4em\relax IEEE, 2014, pp. 478--482.

\bibitem{kazemipour2015robust}
A.~Kazemipour, M.~Wu, and B.~Babadi, ``Robust estimation of self-exciting point
  process models with application to neuronal modeling,'' \emph{arXiv preprint
  arXiv:1507.03955}, 2015.

\bibitem{grenander1958toeplitz}
U.~Grenander and G.~Szeg{\"o}, \emph{Toeplitz forms and their
  applications}.\hskip 1em plus 0.5em minus 0.4em\relax Univ of California
  Press, 1958, vol. 321.

\bibitem{rudzkis1978large}
R.~Rudzkis, ``Large deviations for estimates of spectrum of stationary
  series,'' \emph{Lithuanian Mathematical Journal}, vol.~18, no.~2, pp.
  214--226, 1978.

\bibitem{Negahban}
S.~N. Negahban, P.~Ravikumar, M.~J. Wainwright, and B.~Yu, ``A unified
  framework for high-dimensional analysis of {M-}estimators with decomposable
  regularizers,'' \emph{Statistical Science}, vol.~27, no.~4, pp. 538--557,
  2012.

\bibitem{van_de_geer}
S.~A. van~de Geer, ``On {H}oeffding's inequality for dependent random
  variables,'' in \emph{{Empirical Process Techniques for Dependent Data}},
  H.~Dehling and W.~Philipp, Eds.\hskip 1em plus 0.5em minus 0.4em\relax
  Springer, 2001.

\bibitem{geer2000empirical}
------, \emph{Empirical Processes in {M}-estimation}.\hskip 1em plus 0.5em
  minus 0.4em\relax Cambridge university press, 2000.

\bibitem{macwilliams1977theory}
F.~J. MacWilliams and N.~J.~A. Sloane, \emph{The theory of error correcting
  codes}.\hskip 1em plus 0.5em minus 0.4em\relax Elsevier, 1977, vol.~16.

\bibitem{graham1980lower}
R.~L. Graham and N.~Sloane, ``Lower bounds for constant weight codes,''
  \emph{Information Theory, IEEE Transactions on}, vol.~26, no.~1, pp. 37--43,
  1980.

\bibitem{lehmann1986testing}
E.~L. Lehmann, J.~P. Romano, and G.~Casella, \emph{Testing statistical
  hypotheses}.\hskip 1em plus 0.5em minus 0.4em\relax Wiley New York et al,
  1986, vol. 150.

\end{thebibliography}
}

\end{document}